\documentclass[11pt]{amsart}

\usepackage[english]{babel}
\usepackage{amsmath}
\usepackage{amssymb}
\usepackage{amsthm}
\usepackage{latexsym}
\usepackage{mathtools}
\usepackage{thmtools}
\usepackage{amsfonts}
\usepackage{mathrsfs}
\usepackage{textcomp}
\usepackage{graphicx}
\usepackage{setspace}
\usepackage{nicefrac}
\usepackage{indentfirst}
\usepackage{enumerate}
\usepackage{wasysym}
\usepackage{upgreek}
\usepackage{paralist}
\usepackage{xcolor}
\usepackage{booktabs}
\usepackage{etoolbox}
\usepackage{mathdots}
\usepackage{wrapfig}
\usepackage{floatflt}
\usepackage{tensor} 
\usepackage{parskip}
\usepackage{lscape}
\usepackage{stmaryrd}
\usepackage[enableskew]{youngtab}
\usepackage{ytableau}
\usepackage{pifont}

\usepackage{tikz}
\usetikzlibrary{arrows}
\usetikzlibrary{matrix}
\usetikzlibrary{positioning}
\usetikzlibrary{snakes}
\usetikzlibrary{calc}

\usepackage[all,cmtip]{xy}
\usepackage[labelfont=bf, font={small}]{caption}[2005/07/16]

\usepackage{xcolor}
\definecolor{red}{rgb}{1,0,0}

\newcommand{\vvirg}{ , \dots , }
\newcommand{\ootimes}{ \otimes \cdots \otimes }
\newcommand{\ooplus}{ \oplus \cdots \oplus }
\newcommand{\ttimes}{ \times \cdots \times }

\newcommand{\dotitem}{ \item[$\cdot$]}

\newcommand{\textsum}{{\textstyle \sum}}
\newcommand{\textprod}{{\textstyle \prod}}

\newcommand{\bigboxtimes}{{\scalebox{1.2}{$\boxtimes$}}}

\newcommand{\bfG}{\mathbf{G}}

\newcommand{\bfN}{\mathbf{N}}

\newcommand{\bfX}{\mathbf{X}}
\newcommand{\bfY}{\mathbf{Y}}
\newcommand{\bfZ}{\mathbf{Z}}

\newcommand{\bfe}{\mathbf{e}}

\newcommand{\bfg}{\mathbf{g}}

\newcommand{\bfm}{\mathbf{m}}
\newcommand{\bfn}{\mathbf{n}}

\newcommand{\bfu}{\mathbf{u}}
\newcommand{\bfv}{\mathbf{v}}

\newcommand{\calG}{\mathcal{G}}

\newcommand{\calS}{\mathcal{S}}

\newcommand{\calZ}{\mathcal{Z}}

\newcommand{\bbC}{\mathbb{C}}

\newcommand{\bbN}{\mathbb{N}}

\newcommand{\bbP}{\mathbb{P}}

\newcommand{\frakg}{\mathfrak{g}}
\newcommand{\frakh}{\mathfrak{h}}

\newcommand{\frakp}{\mathfrak{p}}

\newcommand{\frakz}{\mathfrak{z}}

\renewcommand{\phi}{\varphi}

\renewcommand{\theta}{\vartheta}

\renewcommand{\tilde}[1]{\widetilde{#1}}
\renewcommand{\hat}[1]{\widehat{#1}}
\renewcommand{\bar}[1]{\overline{#1}}

\newcommand{\id}{\mathrm{id}}
\newcommand{\Id}{\mathrm{Id}}

\renewcommand{\Im}{\mathrm{Im} \;}

\DeclareMathOperator{\Hom}{Hom}

\DeclareMathOperator{\trace}{trace}

\DeclareMathOperator{\End}{End}

\newcommand{\Flat}{\mathrm{Flat}}

\newcommand{\frakgl}{\mathfrak{gl}}

\newcommand{\fillwidthof}[3][c]
	{%
		\parbox
		{%
			\widthof{#2}%
		}%
		{%
			\ifx#1c%
				\centering#3%
			\else\ifx#1l%
				#3\hfill%
			\else\ifx#1r%
				\hfill#3%
			\fi\fi\fi%
		}%
	}%

\def\mylettrine#1#2 {\lettrine{#1}{#2}\space}

\newcommand{\partinto}[1][]{\smash{\mathord{\mathchoice{%
  \xymatrix@=0.4em@1{%
  \ar@{|-}[rr]_-*--{\scriptstyle #1}
  &*{\phantom{\scriptstyle{#1}}}&}
}{
  \xymatrix@=0.25em@1{%
  \ar@{|-}[rr]_-*--{\scriptstyle #1}
  &*{\phantom{\scriptstyle{#1}}}&}
}{
  \xymatrix@=0.2em@1{%
  \ar@{|-}[rr]_-*--{\scriptscriptstyle #1}
  &*{\phantom{\scriptscriptstyle{#1}}}&}
}{}}}}
\newcommand{\partintonosmash}[1][]{\mathord{\mathchoice{%
  \xymatrix@=0.4em@1{%
  \ar@{|-}[rr]_-*--{\scriptstyle #1}
  &*{\phantom{\scriptstyle{#1}}}&}
}{
  \xymatrix@=0.25em@1{%
  \ar@{|-}[rr]_-*--{\scriptstyle #1}
  &*{\phantom{\scriptstyle{#1}}}&}
}{
  \xymatrix@=0.2em@1{%
  \ar@{|-}[rr]_-*--{\scriptscriptstyle #1}
  &*{\phantom{\scriptscriptstyle{#1}}}&}
}{}}}
\newcommand{\partintostar}[1][]{\smash{\mathord{\mathchoice{%
  \xymatrix@=0.4em@1{%
  \ar@{|-}[rr]_-*--{\scriptstyle #1}^-*--{\scriptstyle \ast}
  &*{\phantom{\scriptstyle{#1}}}&}
}{
  \xymatrix@=0.25em@1{%
  \ar@{|-}[rr]_-*--{\scriptstyle #1}^-*--{\scriptstyle \ast}
  &*{\phantom{\scriptstyle{#1}}}&}
}{
  \xymatrix@=0.2em@1{%
  \ar@{|-}[rr]_-*--{\scriptscriptstyle #1}^-*--{\scriptstyle \ast}
  &*{\phantom{\scriptscriptstyle{#1}}}&}
}{}}}}
\newcommand{\partintostarnosmash}[1][]{\mathord{\mathchoice{%
  \xymatrix@=0.4em@1{%
  \ar@{|-}[rr]_-*--{\scriptstyle #1}^-*--{\scriptstyle \ast}
  &*{\phantom{\scriptstyle{#1}}}&}
}{
  \xymatrix@=0.25em@1{%
  \ar@{|-}[rr]_-*--{\scriptstyle #1}^-*--{\scriptstyle \ast}
  &*{\phantom{\scriptstyle{#1}}}&}
}{
  \xymatrix@=0.2em@1{%
  \ar@{|-}[rr]_-*--{\scriptscriptstyle #1}^-*--{\scriptstyle \ast}
  &*{\phantom{\scriptscriptstyle{#1}}}&}
}{}}}

\usepackage[top=3cm,bottom=3cm,left=3cm,right=3cm]{geometry}
\usepackage{tikz-cd} 

\usepackage{subcaption}

\usepackage[colorlinks=true]{hyperref}
 \hypersetup{
     colorlinks,
     citecolor={blue!50!black},
     linkcolor={red!70!black},
     urlcolor={green!40!black}
 }
 
\DeclareMathOperator{\Stab}{Stab} 
\newcommand{\SL}{\mathrm{SL}}
\newcommand{\GL}{\mathrm{GL}}
\newcommand{\PGL}{\mathrm{PGL}}

\newcommand{\TNS}{\mathcal{T\!N\!S}}

\newcommand{\calTNS}{\mathcal{T\!N\!S}}
\newcommand{\link}{\\ \url{https://fulges.github.io/code/BDG-DimensionTNS.html}}

\title[Dimension of Tensor Network Varieties]{Dimension of Tensor Network Varieties}

\author[A. Bernardi, C. De Lazzari, F. Gesmundo]{Alessandra Bernardi, Claudia De Lazzari, Fulvio Gesmundo}

\address[A. Bernardi, C. De Lazzari]{Dipartimento di Matematica, Universit\`a di Trento, Via Sommarive 14, 38123 Povo (TN), Italy}
\email[Bernardi]{alessandra.bernardi@unitn.it}
\email[De Lazzari]{claudia.delazzari@unitn.it}

\address[F. Gesmundo]{Max Planck Institute for Mathematics in the Sciences, Inselstrasse 22, 04103, Leipzig, Germany; (current) Saarland University, Saarbr\"ucken, Germany}
\email[Gesmundo]{gesmundo@cs.uni-saarland.de}

\newtheorem{lemma}{Lemma}[section]
\newtheorem{conjecture}[lemma]{Conjecture}
\newtheorem{proposition}[lemma]{Proposition}
\newtheorem{theorem}[lemma]{Theorem}
\newtheorem{corollary}[lemma]{Corollary}

\theoremstyle{definition}
\newtheorem{definition}[lemma]{Definition}
\newtheorem{remark}[lemma]{Remark}

\newcommand{\ti}{\tilde{\imath}}

\makeatletter
\@namedef{subjclassname@2020}{%
  \textup{2020} Mathematics Subject Classification}
\makeatother

\begin{document}
\begin{abstract}
The tensor network variety is a variety of tensors associated to a graph and a set of positive integer weights on its edges, called bond dimensions. We determine an upper bound on the dimension of the tensor network variety. A refined upper bound is given in cases relevant for applications such as varieties of matrix product states and projected entangled pairs states. We provide a range (the ``supercritical range'') of the parameters where the upper bound is sharp.
\end{abstract}

\maketitle

\section{Introduction}

Tensor network varieties are varieties of tensors described by the combinatorial structure of a graph. They play a major role in quantum many-body physics, where they are used as a variational ansatz class to describe strongly correlated quantum systems whose entanglement structure is given by the underlying graph. 

The original motivation in quantum physics is the description of quantum spin chains, see \cite{AffKenLieTas:ValenceBondGroundStates,FanNacWer:FinitelyCorrQuantumSpin,OstRom:ThermoLimitDensityMatrixRenorm}. In this setting, it is known that ground states of a local gapped Hamiltonian on $1$-dimensional spin chains are well approximated by matrix product states, which are tensor network states associated to a circular graph \cite{PerVerWolCir:MPSrepresentations,VerMurCir:MPSProjEntPairStaEtc}. We refer to \cite{Orus:PracticalIntroTensorNetworks,SilTscGerJunJasRizMon:TNAntology} for a full description of the subject from the point of view of quantum physics. Methods from differential and complex geometry were introduced in the study of these objects in \cite{HaegMarOsbVer:GeometryMPS} and more recently
some important developments were achieved using methods from algebraic geometry and representation theory \cite{BucBucMic:HackbushConjecture,MicSeyVer:TensorWielandt,GesLanWal:MPSandQuantumMaxFlowMinCut,ChrLucVraWer:TensorNetworkRepGeom,ChrGesStWer:Optimization,HGSHDC:GeomVariationalMethods}. 

Moreover, tensor networks have a role in other areas of applied mathematics. In algebraic complexity theory, the model of computation of algebraic branching program is a ``symmetrized version'' of a tensor network \cite{BeCl:Alg_for_const_reg,DvMaPeYe:Multi_Branching_Programs}. In algebraic statistics, probabilistic graphical models are described as a joint probability distribution of a set of random variables whose correlations factor through the structure of the graph \cite{Lauritzen:GraphicalModels,RobSei:TNSandGraphicalModels}; these models find application in phylogenetics \cite{EriRanStuSull:PhyloAG,AllRho:PhylogeneticIdealsVarsGeneralMarkovModel}. In machine learning, a linear network is essentially a tensor network where the contraction maps are usually precomposed with a nonlinear \emph{activation} function \cite{Ben:LearningDeepArch}.

In this work, we approach the problem of determining the dimension of tensor network varieties, that is the closure of the set of tensors allowing a tensor network representation for a given graph. This provides a measure of how large the set of tensors allowing a certain tensor network representation is, which in turn gives a measure of the expressiveness of the tensor network class. We provide a completely general upper bound in Theorem \ref{theorem: main} and we illustrate how to refine it in cases relevant for applications in Corollary \ref{cor: MPS} and Corollary \ref{cor: PEPS}. In Corollary \ref{cor: supercritical exact}, we give the exact value of the dimension of the tensor network variety in a particular range, where it can be realized as the closure of the orbit of the action of an algebraic group. In Section \ref{section: def and prelim}, we give a complete description of the objects we are going to study. Section \ref{sec: isotropy groups} and Section \ref{section: dimension} are devoted to the proof of the main results. Finally, in Section \ref{sec: small cases}, we further analyze some cases arising from small values of the parameters, and we provide a more precise calculation of their dimension.

\subsection{Main results}
Let $V_1 \vvirg V_d$ be complex vector spaces with $\dim V_i = n_i$ and let $\Gamma$ be a simple graph with vertex set $\bfv(\Gamma)$ of cardinality $d$ and edge set $\bfe(\Gamma)$. The tensor network varieties associated to $\Gamma$ in $V_1 \ootimes V_d$ are irreducible algebraic varieties in $V_1 \ootimes V_d$ depending on a collection of integer weights $\bfm = (m_e : e \in \bfe(\Gamma))$ on the edges of $\Gamma$, called \emph{bond dimensions}. Write $\bfn = (n_1 \vvirg n_d)$ for the local dimensions of the tensor product, and let $\calTNS^\Gamma_{\bfm,\bfn}$ be the tensor network variety associated to $\Gamma$ with bond dimensions $\bfm$ in $V_1 \ootimes V_d$; see Definition \ref{TNS:def}. 

It will be clear from the definitions that if $\bfm$ and $\bfm'$ are two collections of weights such that $m'_e \leq m_e$ for every edge $e \in \bfe(\Gamma)$ then $\calTNS^\Gamma_{\bfm',\bfn} \subseteq \calTNS^{\Gamma}_{\bfm,\bfn}$.

Our main result is the following
\begin{theorem}\label{theorem: main}
 Let $(\Gamma,\bfm,\bfn)$ be a tensor network and let $\calTNS^\Gamma_{\bfm,\bfn}$ be the corresponding tensor network variety. Then
 \begin{align*}
  \dim &\calTNS^\Gamma_{\bfm,\bfn} \leq \\ 
  &\min \left\{ \sum_{v \in \bfv(\Gamma)} (n_v \cdot \textprod_{e \ni v} m_e) - d +1 - \sum_{e \in \bfe(\Gamma)} (m_e^2 -1) + \dim \Stab_{\calG_{\Gamma,\bfm}}(X), \prod_{v \in \bfv(\Gamma)} n_v \right \} .
 \end{align*}
\end{theorem}
In the statement of Theorem \ref{theorem: main}, $\Stab_{\calG_{\Gamma,\bfm}}(X)$ is the stabilizer under the action of the \emph{gauge subgroup} of a generic $d$-tuple of linear maps, whose role will be made clear in Section \ref{section: dimension}. 

The term $\dim \Stab_{\calG_{\Gamma,\bfm}}(X)$ in Theorem \ref{theorem: main} makes the statement not immediate to apply in full generality, as it describes the dimension of the tensor network variety in terms of the dimension of another object which is not immediate to compute. However, as it will be explained in Section \ref{subsec: sharpening}, the value $\dim \Stab_{\calG_{\Gamma,\bfm}}(X)$ can be bounded from above by the dimension of a potentially larger stabilizer which can be computed from the \emph{local structure} of the graph, rather than from its global combinatorics. In fact, a consequence of Proposition \ref{prop: stable if deg high enough} will be that the term $\dim \Stab_{\calG_{\Gamma,\bfm}}(X)$ is trivial in a wide range of cases.

The term $\sum_{e \in \bfe(\Gamma)} (m_e^2 -1)$ is the dimension of the gauge subgroup associated to the tensor network, see Section \ref{sec: gauge section}. The role of this group in the theory of tensor network was known and it is expected that it entirely controls the value of $\dim \calTNS^\Gamma_{\bfm,\bfn}$. In fact, it is expected that in ``most'' cases the exact value of the dimension is 
\begin{equation}\label{eqn: expected value}
\min \left\{ \sum_{v \in \bfv(\Gamma)} (n_v \cdot \textprod_{e \ni v} m_e) - d +1 - \sum_{e \in \bfe(\Gamma)} (m_e^2 -1) , \prod_{v \in \bfv(\Gamma)} n_v \right \} .
\end{equation}
However, in Section \ref{sec: small cases}, we will observe that there are at least some cases where the value in \eqref{eqn: expected value} provides a strict upper bound. Conjecture \ref{conjecture: bond 2} predicts that the value in \eqref{eqn: expected value} is indeed the dimension of the tensor network variety in the case of matrix product states of bond dimension two, with the only exceptions classified in Section \ref{sec: small cases}.

Particularly relevant in the study of quantum many-body systems are \emph{lattice graphs} for which we provide some examples in Figure \ref{fig: MPS and PEPS}. In the physics literature, elements of the tensor network variety associated to a path or a cycle are called matrix product states (MPS), respectively with open or periodic boundary condition; elements of the tensor network variety associated to a grid (or more generally a two-dimensional lattice), either placed on a plane or on a torus, are called projected entangled pair states (PEPS), with open or periodic boundary conditions respectively.

In the case of matrix product states with open boundary conditions, a complete result regarding the dimension of the tensor network variety is given in \cite[Thm 14]{HaegMarOsbVer:GeometryMPS}. In the language of Theorem \ref{theorem: main}, setting $\bfm = (m_1 \vvirg m_{d-1})$ and $\bfn = (n_1 \vvirg n_d)$ to be the collections of bond dimensions and of local dimensions on the path $P_d$, the result of \cite{HaegMarOsbVer:GeometryMPS} is (formally setting $m_0 = m_d = 1$)
\[
 \dim \calTNS^{P_d}_{\bfm,\bfn} = \min \left\{ \textsum_{i =1}^d  n_i  m_{i-1}m_{i}  - \textsum_{j = 1}^{d-1} m_i^2 , \textprod_{i=1}^d n_i\right\};
\]
this coincides with the expected value for the dimension in \eqref{eqn: expected value}. 

We refer to \cite{PerVerWolCir:MPSrepresentations,ChrLucVraWer:TensorNetworkRepGeom} for the details on the construction of MPS, PEPS and other related entanglement structures and for their physical interpretation. 

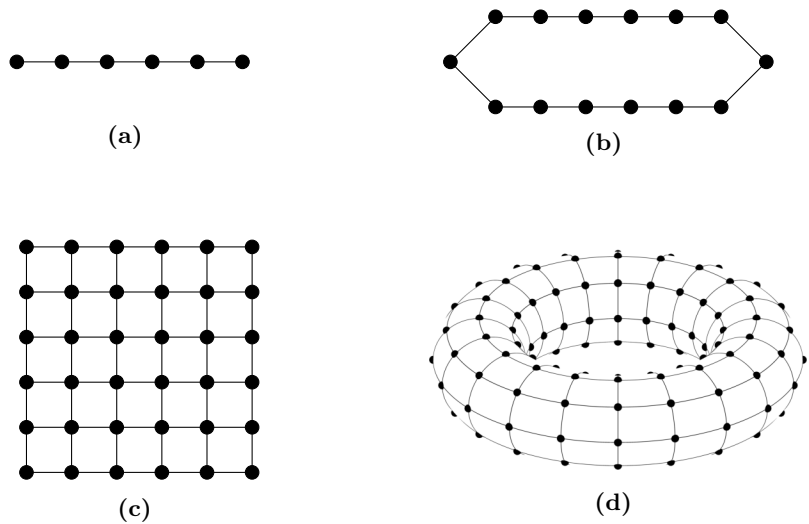
\begin{figure}
\begin{subfigure}{.4\textwidth}
\centering
\begin{tikzpicture}[scale=.6]
\path [use as bounding box] (-1,-1) rectangle (6,1);
\foreach \x in {0,...,5}
	{
	\draw[fill=black] (\x,0) circle (.15cm);
	}
\draw(0,0)--(5,0);
\end{tikzpicture}
\subcaption{~}
 \end{subfigure}
\begin{subfigure}{.4\textwidth}
\centering
\begin{tikzpicture}[scale=.6]
\foreach \x in {0,...,5}
	{
	\draw[fill=black] (\x,0) circle (.15cm);
	\draw[fill=black] (\x,2) circle (.15cm);
	}
\draw (0,0)--(5,0);
\draw (0,2)--(5,2);
\draw (0,0)--(-1,1);
\draw (0,2)--(-1,1);
\draw (5,0)--(6,1);
\draw (5,2)--(6,1);
\draw[fill=black] (-1,1) circle (.15cm);
\draw[fill=black] (6,1) circle (.15cm);
 \end{tikzpicture}
\subcaption{~}
 \end{subfigure}\\~
~ \vspace{1cm}\\~
\begin{subfigure}{.4\textwidth}
\centering
\begin{tikzpicture}[scale=.6]
\foreach \x in {0,...,5}
    \foreach \y in {0,...,5} 
	{
	\draw[fill=black] (\x,\y) circle (.15cm);
	}
\foreach \x in {0,...,5}
	{
	\draw (\x,0)--(\x,5);
	}
\foreach \y in {0,...,5}
	{
	\draw (0,\y)--(5,\y);
	}
\end{tikzpicture} 
\subcaption{~}
 \end{subfigure}
\begin{subfigure}{.4\textwidth}
\centering
\includegraphics{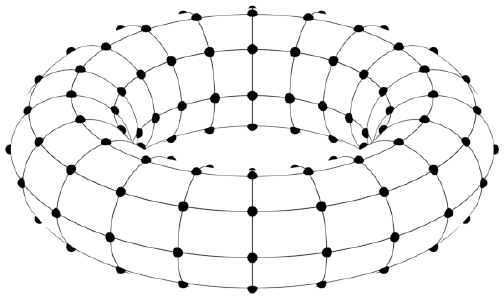}
\subcaption{~}
\end{subfigure}
\caption{Examples of lattice graphs: MPS with open (a) and periodic (b) boundary conditions; PEPS with open (c) and periodic (d) boundary conditions.}\label{fig: MPS and PEPS}
\end{figure}

We state the following corollaries of Theorem \ref{theorem: main} in the case of constant bond dimension and constant local dimension. It will be clear from the discussion of Section \ref{section: dimension} that these hypotheses can be relaxed, but we state them in this restricted range for the sake of presentation.

\begin{corollary}\label{cor: MPS}
Let $(C_d, m,n)$ be the tensor network on the cycle graph on $d$ vertices with constant bond dimension $m$ and constant local dimension $n$. Then 
\[
  \dim \calTNS^{C_d}_{m,n} \leq \min \{ d (n -1) m^2  +1 , n^d\}. 
\]
\end{corollary}

\begin{corollary}\label{cor: PEPS}
Let $\Gamma$ be a graph on $d$ vertices such that all vertices of $\Gamma$ have degree at least $3$. Let $(\Gamma,m,n)$ be the tensor network on $\Gamma$ with constant bond dimension $m$ and constant local dimension $n$. Then 
\[
  \dim \calTNS^{\Gamma}_{m,n} \leq \min \left\{ \sum_{v \in \bfv(\Gamma)} n m^{\deg(v)}  - d +1 - \sum_{e \in \bfe(\Gamma)} (m^2 -1) , n^d \right \} .
\]
\end{corollary}

The equality in \eqref{eqn: expected value} is attained in the \emph{supercritical range}, defined in Section \ref{section: dimension}.
 \begin{corollary}\label{cor: supercritical exact}
  Let $(\Gamma,\bfm,\bfn)$ be a supercritical tensor network. Then 
  \[
   \dim \calTNS^\Gamma_{\bfm,\bfn} = \min \left\{ \sum_{v \in \bfv(\Gamma)} (n_v \cdot \textprod_{e \ni v} m_e) - d +1 - \sum_{e \in \bfe(\Gamma)} (m_e^2 -1) , \prod_{v \in \bfv(\Gamma)} n_v \right \} .
  \]
 \end{corollary}

 \subsection{State of the art and related work}
  
Tensor network varieties appeared in \cite{LanQiYe:GeomTensorNetwork}, where a number of basic geometric questions were answered, providing several insights. In \cite{LimYe:TensorNetworkRanks}, a comparison between tensor network varieties corresponding to different underlying graphs is proposed. The problem of dimension is also addressed: in particular, Theorem 7.4 and Theorem 9.1 in \cite{LimYe:TensorNetworkRanks} correspond to the particular cases of Corollary \ref{cor: supercritical exact} where the underlying graph is a path or a cycle respectively; in this case, the result follows also from Proposition 2.9 in \cite{Ges:Geometry_of_IMM}.   

It would be interesting to have a general understanding of lower bounds for the dimension of tensor network varieties, but this is a challenging problem. In Section \ref{section: def and prelim}, we provide a parametrization of an open subset of the tensor network variety: hence determining lower bounds on the dimension can be reduced to determining lower bounds on the rank of the differential of the parametrization at a point; however, determining a suitable point, and computing such rank is non-trivial.

An indirect method to determine lower bounds on the dimension of tensor network varieties consists in determining subvarieties of known dimension contained in it. The result of Corollary 4.2 in \cite{LimYe:TensorNetworkRanks} would give a lower bound of this form, whenever the dimension of the $r$-th secant variety of the Segre variety of rank one tensors is known.  The result of Corollary 4.2 of \cite{LimYe:TensorNetworkRanks} can be improved using lower bounds on the \emph{border subrank} of the graph tensors introduced in Section \ref{section: def and prelim} below. We only mention a result in this direction which follows from Theorem 6.6 in \cite{Str:RelativeBilComplMatMult}. For $\bfn = (n_1,n_2,n_3)$, write $\sigma_{r,\bfn}$ for the $r$-th secant variety of the variety of rank one tensors in $\bbC^{n_1} \otimes \bbC^{n_2} \otimes \bbC^{n_3}$. Then, for $\bfm= (m_1 , m_2,m_3)$ with $m_1 \leq m_2 \leq m_3$, we have 
\[
 \sigma_{r,\bfn} \subseteq \calTNS^{C_3} _{\bfm,\bfn}
\]
for every $r \leq m_1m_2 - \left\lfloor \frac{(m_1+m_2 -m_3)^2}{4}\right\rfloor$. In particular, if $m:= m_1 = m_2 =m_3$, we get $\sigma_{r,\bfn} \subseteq \calTNS^{C_3}_{\bfm,\bfn}$ for $r \leq \lceil 3/4 m^2\rceil$; moreover, if $n = n_1=n_2=n_3$, \cite{Lick:Typical_tensor_rank} provides $\dim \sigma_{r,\bfn} = \min\{ r (3n-2), n^3\}$, with the only exception $(r,n) = (4,3)$ where $\dim \sigma_{4,3} = 26$; we deduce
\[
\dim  \calTNS^{C_3}_{\bfm,\bfn} \geq \min\{ r (3n-2), n^3\}
\]
for $r = \lceil 3/4 m^2\rceil$. We do not expect this method to give a sharp bound except in trivial cases. 

Indeed, we expect the upper bound of Theorem \ref{theorem: main} to give the exact value of the dimension in ``most'' cases, in a way similar to the Alexander-Hirschowitz Theorem for secant varieties of Veronese varieties \cite{AlHir:Poly_interpolation_in_several_variables}.

An upper bound analogous to the one of Theorem \ref{theorem: main} is proposed for translation invariant matrix product states in \cite[Conjecture 2.14]{CzaMicSey:UniformMPS}. One can verify that this value coincides with the dimension of the variety of translation invariant matrix product states for a number of small parameter values; in particular, there are no known exceptions in the translation invariant setting, in contrast with the exceptions that we will highlight in Section \ref{sec: small cases}.

\section{Definitions and preliminaries}\label{section: def and prelim}

We introduce tensor network varieties via the language of graph tensors, following \cite{VrChr:EntDistGHZShares,ChrVraZui:AsyRankGraph}. In this work, we restrict to simple graphs; the theory generalizes to more general notions of graphs and we refer to \cite{ChrLucVraWer:TensorNetworkRepGeom,ChrGesMicZui:BorRankNonAddHigher} for the definitions and the basics in the general setting.

Given tensors $T \in V_1 \ootimes V_d$ and $S \in V'_1 \ootimes V'_d$, the \emph{Kronecker product} of $T$ and $S$, denoted $T \boxtimes S$, is the element $T \otimes S$ regarded as a tensor on $d$ factors
\[
 T \boxtimes S \in (V_1 \otimes V'_1) \ootimes (V_d \otimes V'_d).
\]
Given a tensor $T \in V_1 \ootimes V_d$, for every subset $I \subseteq \{ 1 \vvirg d\}$, $T$ defines a linear map $T_I : \bigotimes _{i \in I} V_i^* \to \bigotimes_{i' \notin I} V_{i'}$ called \emph{flattening map} associated to $I$. We say that $T$ is \emph{concise} if all the flattening maps $T_{i} : V_i^* \to \bigotimes _{i' \neq i} V_{i'}$ are injective.

A \emph{simple graph} is an undirected graph with no loops and no multiple edges. Let $\Gamma = (\bfv(\Gamma), \bfe(\Gamma))$ be a simple graph, with vertex set $\bfv(\Gamma) = \{ 1 \vvirg d\}$ and edge set $\bfe(\Gamma)=\{ e_1 \vvirg e_R\}$. A collection of \emph{bond dimensions} is a set of weights $\bfm = (m_e : e \in \bfe(\Gamma))$ on the edges of $\Gamma$. Given a collection of bond dimensions $\bfm$, define the graph tensor associated to $\Gamma$ as follows. For an edge $e = \{ i_1,i_2\}$, let 
\[
\bfu_{(e)}(m_e) = \sum_{j =1}^{m_e} v^{(i_1)}_j \otimes v^{(i_2)}_j \otimes \bigotimes_{i \neq i_1,i_2} v_0^{(i)} \in \bbC^{m_e} \otimes \bbC^{m_e} \otimes \bbC^1 \ootimes \bbC^1
\]
where for $p=1,2$, $\{ v^{(i_p)}_j : j =1 \vvirg m_e\}$ are bases of a copy of $\bbC^{m_e}$ and $v_0^{(i)}$ is a generator of $\bbC^1$ for $i\neq i_1,i_2$. The superscripts indicate the ordering of the tensor factors. 

The graph tensor associated to a graph $\Gamma$ with bond dimensions $\bfm$ is 
\begin{equation}\label{T:Gamma:m}
 T(\Gamma,\bfm) = \bigboxtimes _{e \in \bfe(\Gamma)} \bfu_{(e)}(m_e);
 \end{equation}
this is a tensor of order $d$ whose $i$-th factor has a local structure $W_i = \bigotimes_{e \ni i} \bbC^{m_e}$. In coordinates, we pictorially imagine the graph tensor $T(\Gamma,\bfm)$ as the tensor product of identity matrices $\Id_{m_e} \in \bbC^{m_e} \otimes \bbC^{m_e}$ for $e \in \bfe(\Gamma)$ laying on the edges of the graph; this product is regarded as a tensor of order $d$ where the $i$-th factor is the product of the spaces $\bbC^{m_e}$ incident to vertex $i$. Note that from this point of view one of the two copies of $\bbC^{m_e}$ is identified with its dual space ${\bbC^{m_e}}^*$, see Figure \ref{fig: kroneckering on triangle}.

\begin{figure}
\begin{minipage}{.4\textwidth}
\begin{tikzpicture}
\path [use as bounding box] (-1,-1) rectangle (3,3);
\draw[fill=black] (-.5,0) circle (.15cm);
\draw[fill=black] (0,-.5) circle (.15cm);
\draw[fill=black] (2,-.5) circle (.15cm);
\draw[fill=black] (-.5,2) circle (.15cm);
\draw[fill=black] (2.5,0) circle (.15cm);
\draw[fill=black] (0,2.5) circle (.15cm);
\draw (-.5,0)--(-.5,2);
\draw (0,-.5)--(2,-.5);
\draw (0,2.5)--(2.5,0);
\draw (-.5,1) node[anchor=east] {$\bfu_{(31)}(m_{31})$};
\draw (1.2,1.2) node[anchor=south west] {$\bfu_{(23)}(m_{23})$};
\draw (1,-.5) node[anchor=north] {$\bfu_{(12)}(m_{12})$};
\draw (4.75,1) node {$\bigboxtimes$};
\draw[->] (4,.7)--(5.5,.7);
\end{tikzpicture}
\end{minipage}
\begin{minipage}{.4\textwidth}
\begin{tikzpicture}[scale=1.5]
\path [use as bounding box] (-1,-1) rectangle (3,3);
\draw (0,2)-- (0,0);
\draw (0,0)-- (2,0);
\draw (0,2)-- (2,0);
\draw (1,0) node[anchor=north] {$m_{12}$};
\draw (1,1) node[anchor=south west] {$m_{23}$};
\draw (0,1) node[anchor=east] {$m_{31}$};
\draw[fill=black] (0,0) circle (.1cm);
\draw[fill=black] (2,0) circle (.1cm);
\draw[fill=black] (0,2) circle (.1cm);
\draw (2.5,.5) node {$=T(\Gamma,\bfm)$};
\end{tikzpicture}
\end{minipage}
\caption{Pictorial representation of the construction of the graph tensor $T(\Gamma,\bfm)$ on the triangular graph: $T(\Gamma,\bfm)$ is the tensor product of the three identity matrices $\bfu_e(m_e)$, regarded as a tensor on three factors.}\label{fig: kroneckering on triangle}
\end{figure}
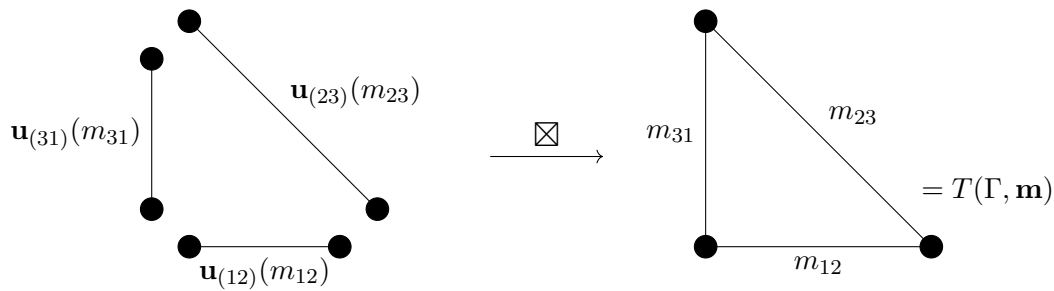

\begin{remark}\label{rmk: dimension 1}
Let $\Gamma$ and $\Gamma'$ be two graphs on the same set of vertices and with $\bfe(\Gamma) = \bfe(\Gamma') \cup \{ e\}$. In other words, $\Gamma'$ is the graph obtained from $\Gamma$ after removing the edge $e$. Let $\bfm$ be a collection of bond dimensions on $\Gamma$ and let $\bfm'$ be the collection $\bfm$ restricted to $\Gamma'$. It is clear from the definitions that if $m_e = 1$ then $T(\Gamma,\bfm) = T(\Gamma',\bfm')$ because in this case $\bfu_{(e)}(m_e)$ is a decomposable tensor hence $T \boxtimes \bfu_{(e)}(m_e)= T$ for every tensor $T$.
\end{remark}

Remark \ref{rmk: dimension 1} guarantees that up to modifying the underlying graph, one can always assume $m_e \geq 2$. 

Let $n_{i} \in \bbN$ be integers associated to the vertices of $\Gamma$ and let $V_i = \bbC^{n_i}$. Write $\bfn = (n_i)_{i = 1 \vvirg d}$ for the $d$-uple of dimensions of the vector spaces $V_i$; we say that $\bfn$ are the \emph{local dimensions} associated to $\Gamma$. A triple $(\Gamma, \bfm, \bfn)$ consisting of a simple graph, a collection of bond dimensions and a collection of local dimensions is a \emph{tensor network}. A tensor network naturally provides the following algebraic variety. 
\begin{definition}[Tensor Network Variety]\label{TNS:def}
The \emph{tensor network variety} in $V_1 \ootimes V_d$ associated to the tensor network $(\Gamma, \bfm, \bfn)$ is 
\begin{align*}
\TNS^\Gamma_{\bfm, \bfn} = \overline{ \biggl\{T \in V_1\ootimes V_d : T  = (X_1 \ootimes X_d) \cdot T(\Gamma,\bfm), X_j \in \Hom(W_j,V_j) \biggr\}},
\end{align*}
where the closure can be taken equivalently in the Euclidean or the Zariski topology.
\end{definition}
The set $\TNS^\Gamma_{\bfm, \bfn}$ is an irreducible algebraic variety. Moreover, if $\bfm$ and $\bfm'$ are two collections of bond dimensions on $\Gamma$ such that $m'_e \leq m_e$ for every edge $e \in \bfe(\Gamma)$, we have $\TNS^\Gamma_{\bfm', \bfn} \subseteq \TNS^\Gamma_{\bfm, \bfn}$.

It is known that if the graph $\Gamma$ is a tree, then the closure in the definition of $\calTNS_{\bfm,\bfn}^\Gamma$ is not needed, but if $\Gamma$ contains cycles then there are examples where it is necessary \cite{LanQiYe:GeomTensorNetwork,ChrLucVraWer:TensorNetworkRepGeom,BarLuFri:ClosednessGeometryTNS}.

The definition of $\calTNS^{\Gamma}_{\bfm,\bfn}$ provides a natural parameterization of a dense subset given by 
\begin{align*}
	\hat{\Phi} : \Hom(W_1, V_1) \ooplus \Hom(W_d,V_d) & \to V_1 \ootimes V_d ,\\
	(X_1, \dots , X_d) & \mapsto (X_1 \ootimes X_d)\cdot T(\Gamma,\bfm).
\end{align*}
Let $\calTNS^{\Gamma \circ}_{\bfm,\bfn}$ be the image of the map $\hat{\Phi}$. The set $\calTNS^{\Gamma \circ}_{\bfm,\bfn}$ is often the object of interest in applications, as it coincides exactly with the set of tensors which have a tensor network representation with the given parameters. Since we are interested in the dimension of these objects, as we employ methods from algebraic geometry, we consider the algebraic variety obtained taking the closure. The map $\hat{\Phi}$ factors as follows:
\begin{equation*}
	\begin{tikzcd}
		\bigoplus_{i=1}^d \Hom(W_i, V_i) \arrow[r, "\mu"]\arrow[dr, "\hat\Phi"'] & \Hom(W_1 \ootimes W_d, V_1 \ootimes V_d)\arrow[d, "\bar\Phi"] \\
		&  V_1 \ootimes V_d
	\end{tikzcd}
\end{equation*}
where $\mu$ is the $d$-linear map defined as $\mu(X_1, \dots , X_d)=X_1 \ootimes X_d$. Denote the image of the map $\mu$ by $\Hom ( W_1 \vvirg W_d; V_1 \vvirg V_d)$. It is the cone over the Segre embedding of $\bbP(\Hom(W_1, V_1)) \ttimes \bbP(\Hom(W_d,V_d))$ in $\bigotimes_1^d \Hom(W_i,V_i) = \Hom(W_1 \ootimes W_d, V_1 \ootimes V_d) $ and its affine dimension is 
\[
\dim \Hom(W_1,\dots, W_d; V_1,\dots,V_d)=\sum_{1=1}^d \dim(\Hom(W_i,V_i))-d+1.
\]
The map $\bar{\Phi}$ is simply the evaluation at the graph tensor; therefore the restriction of $\bar{\Phi}$ to the subvariety $\Hom(W_1 \vvirg W_d, V_1 \vvirg V_d)$ provides a parameterization of $\calTNS^{\Gamma \circ}_{\bfm,\bfn}$; denote this restriction by
\[
 \Phi: \Hom(W_1 \vvirg W_d, V_1 \vvirg V_d) \to V_1 \ootimes V_d.
\]

The \emph{dimension} of an irreducible algebraic variety is defined as the dimension of its tangent space at a smooth point. We refer to \cite[Ch. 3]{Shaf:BasicAlgGeom1} for the basic properties of dimension. The Theorem of Dimension of the Fibers \cite[Thm. 1.25]{Shaf:BasicAlgGeom1} provides
\begin{equation}\label{eqn: dimension of fibers}
\dim \calTNS^{\Gamma}_{\bfm,\bfn} = \dim \left[\Hom(W_1,\dots, W_d; V_1,\dots,V_d) \right] - \dim \Phi^{-1}(T)
\end{equation}
where $T$ is a generic tensor in the image of $\Phi$.

The goal of the rest of the paper is to determine the value $\dim \Phi^{-1}(T)$ which, via \eqref{eqn: dimension of fibers}, gives the value of $\dim \calTNS^{\Gamma}_{\bfm, \bfn}$. Determining the exact value $\dim \Phi^{-1}(T)$ is hard in general. We focus on lower bounds, which via \eqref{eqn: dimension of fibers} provide upper bounds for $\dim \calTNS^\Gamma_{\bfm,\bfn}$. This is done by determining the dimension of stabilizer of the graph tensor under the action of $\GL(W_1) \ttimes \GL(W_d)$ and then showing that $ \Phi^{-1}(T)$ contains orbits under the action of such stabilizer; lower bounds on the dimension of such orbit gives a lower bound on $\dim \Phi^{-1}(T)$.

\section{Isotropy of tensors: the gauge subgroup}\label{sec: isotropy groups}

In this section, we determine the dimension of the isotropy group of graph tensors. First, we provide some preliminary results on the stabilizer of a tensor under the action of the general linear groups acting on the tensor factors; we then introduce the gauge subgroup of a tensor network and we prove that it coincides with the connected component of the identity of the isotropy group of the corresponding graph tensor.

\subsection{Isotropy groups of tensors}\label{subsec: isotropy groups}
Given vector spaces $V_1 \vvirg V_d$, consider the natural action of the group $\GL(V_1) \ttimes \GL(V_d)$ on $V_1 \ootimes V_d$. This defines a group homomorphism 
\begin{align*}
	\GL(V_1) \ttimes \GL(V_d) &\to \GL(V_1 \ootimes V_d) \\
	(g_1 \vvirg g_d) &\mapsto g_1 \ootimes g_d
\end{align*}
whose kernel is the central subgroup $Z_{V_1 \ootimes V_d} = \{ (\lambda_1 \Id_{V_1} \vvirg \lambda_d \Id_{V_d}) : \lambda_1 \cdots \lambda_d = 1\}$. Therefore, the group $G(V_1 \vvirg V_d) := \GL(V_1) \ttimes \GL(V_d) / Z_{V_1 \ootimes V_d}$ can be identified naturally with a subgroup of $\GL(V_1\ootimes V_d)$ acting faithfully on $V_1 \ootimes V_d$. The elements of $G(V_1 \vvirg V_d)$ will be denoted as tensor products $g_1 \ootimes g_d$ for $g_j \in \GL(V_j)$.

The corresponding Lie algebra action defines a Lie algebra homomorphism 
\begin{align*}
	\frakgl(V_1) \ooplus \frakgl(V_d) &\to \frakgl(V_1 \ootimes V_d) \\
	(X_1 \vvirg X_d) &\mapsto X_1 \otimes \Id_{V_2} \ootimes \Id_{V_d} + \cdots +\Id_{V_1} \ootimes \Id_{V_{d-1}} \otimes X_d ,
\end{align*}
whose kernel is the central algebra $\frakz_{V_1 \ootimes V_d} = \{  (x_1 \Id_{V_1} \vvirg x_d\Id_{V_d}) : x_1 + \cdots + x_d = 0\}$. Hence, the Lie algebra $\frakg(V_1 \vvirg V_d) := \frakgl(V_1) \ooplus \frakgl(V_d) / \frakz_{V_1 \ootimes V_d}$ is a subalgebra of $\frakgl(V_1 \ootimes V_d)$ and coincides with the Lie algebra of $G(V_1 \vvirg V_d)$. With abuse of notation, denote the elements of $\frakg(V_1 \vvirg V_d)$ as $d$-tuples $\bfX = (X_1 \vvirg X_d)$ with $X_j \in \frakgl(V_j)$ with the understanding that $\bfX$ is identified with its image in $\frakg(V_1 \vvirg V_d)$.

\begin{definition}
	Let $T \in V_1 \ootimes V_d$ be a tensor. The isotropy group of $T$, denoted $G_T$, is the stabilizer of $T$ under the action of $G(V_1 \vvirg V_d)$:
	\[
	G_T = \{ g_1 \ootimes g_d \in G(V_1 \vvirg V_d) : g_1 \ootimes g_d (T) = T\}.
	\]
\end{definition}
The group $G_T$ is algebraic and in general it is union of finitely many connected (irreducible) components. Let $G_T^\circ$ denote the connected component containing the identity: $G_T^\circ$ is normal in $G_T$ and $\dim G_T = \dim G_T^\circ$, see, e.g., \cite[Lemma 2.1]{Ges:Geometry_of_IMM}. 

The isotropy Lie algebra of $T$, denoted $\frakg_T$, is the Lie algebra of the group $G_T$, or equivalently the one of $G_T^\circ$; it is the subalgebra of $\frakg(V_1 \vvirg V_d)$ which annihilates $T$ under the Lie algebra action induced by $\frakgl(V_1) \ooplus \frakgl(V_d)$ \cite[Sec. 1.2]{Pro:LieGroups} 
\[
\frakg_T = \{ \bfX = (X_1 \vvirg X_d) \in \frakg(V_1 \vvirg V_d) : \bfX . T = 0\},
\]
where $\bfX.T = \sum_1^d \Id_{V_1} \ootimes X_k \ootimes \Id_{V_d} (T)$ denotes the image via the Lie algebra action. We have $\dim \frakg_T = \dim G_T^\circ = \dim G_T$.

The dimension of the orbit-closure of a tensor $T$ is therefore given by
\[ 
\dim (G(V_1 \vvirg V_d) \cdot T ) = \left[ \textsum_{i=1}^d (\dim V_i)^2- d+1 \right]  - \dim \frakg_T .
\]

We prove preliminary results on isotropy Lie algebras of tensors of higher order. Lemma \ref{lemma: isotropy of non-concise tensor} is classical and we record it here for the reader's convenience. Lemma \ref{lemma: partial lie algebra} concerns the intersection of $\frakg_T$ with the subalgebra of $\frakg(V_1 \vvirg V_d)$ consisting of elements acting only on a subset of the tensor factors; this will be used in a reduction in the proof of Theorem \ref{thm: isotropy group when kron with star}. 

We first record an immediate linear algebra fact.
\begin{lemma}\label{lemma: intersection vs direct sum}
	Let $V$ be a vector space and let $A,B_1, \dots ,B_N$ be subspaces of $V$ such that there exists a subspace $B$ with the property that $A \cap B = \{0\}$ and $B_j \subseteq B$ for every $j=1,\dots ,N$. Then $\bigcap_{j} (A \oplus B_j) = A \oplus \bigcap_j  B_j$.
\end{lemma}

The following result is classical and follows for instance from \cite[Section 1.1]{Bri:Flags}.
\begin{lemma}\label{lemma: isotropy of non-concise tensor}
	Let $T \in V_1 \ootimes V_d$ be a non-concise tensor. Let $V'_i\subseteq V_i$ be subspaces such that $T \in V'_1 \ootimes V'_d$ is concise. Write $\frakh_T$ for the isotropy Lie algebra of $T$ in $\frakg(V'_1 \vvirg V'_d)$ (regarded as a subalgebra of $\frakg(V_1 \vvirg V_d)$) and $\frakg_T$ for the isotropy Lie algebra of $T$ in $\frakg(V_1 \vvirg V_d)$. Then
	\[
	\frakg_T = \frakh_T \oplus \frakp
	\]
	where $\frakp \subseteq \frakg(V_1 \vvirg V_d)$ is the Lie algebra which annihilates the subspace $V'_1 \ootimes V'_d$, that is the algebra generated by $\bigoplus_{i=1}^d ({V'_i}^\perp \otimes V_i ) \subseteq \frakgl(V_1) \ooplus \frakgl(V_d)$.
\end{lemma}

\begin{lemma}\label{lemma: partial lie algebra}
	Let $T \in V_1 \ootimes V_d$. For $I \subseteq \{ 1 \vvirg d\}$, let $ F _T:=T_{I^c}: \bigotimes_{j \in I^c}V_{j}^* \to \bigotimes_{i \in I} V_{i}$ be the flattening map of $T$ corresponding to the subset $I$. Then
	\begin{equation}\label{eqn: partial lie algebra intersection}
	\frakg_T \cap \frakg(V_i : i \in I) = \bigcap_{S \in \Im F_T} \frakg_S.
	\end{equation}
	In particular, if $T$ is concise, $\frakg_T \cap \frakgl(V_j) = 0$ for every $j$.
\end{lemma}
\begin{proof}
	Let $k = |I|$; after possibly reordering the factors, assume $I = \{1 \vvirg k\}$. 
	
	For $\bfX \in \frakg(V_1 \vvirg V_d)$, write $\bfX = (\bfX_1, \bfX_2)$ with $\bfX_1 = (X_1 \vvirg X_k)$ and $\bfX_2 = (X_{k+1} \vvirg X_d)$. Let $\bfX.T$ be the image of $T$ via the action of $\bfX$ and let $F_{\bfX.T}$ be the corresponding flattening map.  By Leibniz's rule, given an element $S' \in V_{k+1}^*\ootimes V_d^*$, $F_{\bfX.T}$ is characterized by the expression 
	\[
	F_{\bfX.T} (S') = F_T ( \bfX_2. S' ) + \bfX_1.F_T(S'),
	\]
	where $\bfX_2$ acts on $V_{k+1}^* \ootimes V_d^*$, $\bfX_1$ acts on $V_1 \ootimes V_k$.
	
	Now, let $\bfX \in \frakg_T \cap \left( \frakgl(V_1) \ooplus \frakgl(V_k)\right)$. Hence, $\bfX = (\bfX_1, {\bf0})$ and $\bfX.T = 0$. Therefore $0 = F_{\bfX.T} (S') = \bfX_1.F_T(S')$, showing $\bfX_1 \in \frakg_S$ for every $S \in \Im F_T$.
	
	Conversely, let $\bfX_1 \in \bigcap_{S \in \Im F_T} \frakg_S$. Let $S_1 \vvirg S_N \in \Im F_T$ be a set of generators and write $T = \sum_{i=1}^N S_i \otimes P_i$ for some $P_i \in V_{k+1} \ootimes V_d$. Let $\bfX = (\bfX_1,{\bf0})$. Then
	\[
	\bfX.T = \sum_{i=1}^N (\bfX_1.S_i) \otimes P_i + \sum_{i=1}^N S_i \otimes {\bf0}.P_i = \sum_{i=1}^N (\bfX_1.S_i) \otimes P_i = 0
	\]
	showing $\bfX \in \frakg_T$.
	
	This concludes the proof of \eqref{eqn: partial lie algebra intersection}. 
	
	The last claim follows by taking $I = \{j\}$: if $T$ is concise, then $F_T$ is surjective and therefore $\bigcap _{S \in \Im F_T} \frakg_S = \bigcap_{v \in V_j} \frakg_{v} = 0$. \end{proof}

By linearity the intersection in Lemma \ref{lemma: partial lie algebra} can be restricted to a basis of the image of the flattening map $\Im F_T$, as it is clear from the proof.

\subsection{Gauge subgroup}\label{sec: gauge section}
Let $\Gamma$ be a graph and $\bfm = (m_e : e \in \bfe(\Gamma))$ a collection of bond dimensions. Let $T = T(\Gamma,\bfm) \in W_1 \ootimes W_d$ be the associated graph tensor. Fix an edge $e = \{ i_1,i_2\} \in \bfe(\Gamma)$: by definition of $T(\Gamma, \bfm)$ there exist vector spaces $U_e$, $W'_{i_1}, W'_{i_2}$ such that $W_{i_1} = U_e \otimes W'_{i_1}$, and $W_{i_2} = U_e^* \otimes W'_{i_2}$ where $\dim U_e = m_e$ and the tensor product structure depends on the local structure at the vertices $i_1$ and $i_2$. The group $\GL(U_e) \times \GL(U_e^*)$ acts on the factor $U_e \otimes U_e^*$ of $W_{i_1} \otimes W_{i_2}$ with kernel the central subgroup $Z_e = \{ (\lambda \Id_{U_e} , \lambda^{-1} \Id_{U_e^*}) : \lambda \in \bbC^*\}$. 

This defines a homomorphism 
\[
\Psi_e : (\GL(U_e) \times \GL(U_e^*))/Z_e \to \GL(W_{k_1} \otimes W_{k_2}) \to G(W_k : {k \in \bfv(\Gamma)}).
\]
As $e$ varies among the edges of $\Gamma$, the images of the different $\Psi_e$'s commute and therefore they induce a homomorphism
\[
 \Psi : \bigtimes_{e \in \bfe(\Gamma)} (\GL(U_e) \times \GL(U_e^*)) / Z_e \to G(W_k : k \in \bfv(\Gamma)),
\]
which turns out to be injective. Regrouping the factors, we can write
\[
 \Im (\Psi) = \left[ \bigtimes_{v \in \bfv(\Gamma)} H_v  \right] / \left[  \bigtimes_{e \in \bfe(\Gamma)} Z_e \right]
\]
where $H_v = \bigtimes_{v \ni e} \GL_{m_e}$; here $\GL_{m_e}$ is $\GL(U_e)$ or $\GL(U_e^*)$ depending on whether $U_e$ or $U_e^*$ is the tensor factor appearing in $W_v$. With abuse of notation, we will denote by $H_v$ the quotient $\left\langle H_v , \bigl[ \bigtimes_{e \in \bfe(\Gamma)} Z_e \bigr] \right\rangle / \bigl[ \bigtimes_{e \in \bfe(\Gamma)} Z_e \bigr] \subseteq \Im (\Psi)$ as well, where for subgroups $H,K$, one denotes by $\langle H,K\rangle$ the subgroup generated by $H$ and $K$.

Let $\GL^{\Delta}_{m_e} \subseteq \GL(U_e) \times \GL(U_e^*)$ be the subgroup lying ``diagonally'', that is 
\[
\GL^{\Delta}_{m_e} = \{ (A,{A^{-1}}^T) \in \GL(U_e) \times \GL(U_e^*) : A \in \GL(U_e)\};
\]
its image under the homomorphism $\Psi_e$ is a copy of $\PGL_{m_e} \subseteq  G(W_1 \vvirg W_d)$ called \emph{gauge subgroup} on the edge $e$. The following is immediate from the definitions:

\begin{lemma}\label{lemma: gauge stabilizes}
 The gauge subgroup $\PGL_{m_e} \subseteq G(W_1 \vvirg W_d)$ stabilizes $T(\Gamma,\bfm)$.
\end{lemma}
\begin{proof}
Let $e = \{ i_1,i_2\}$, so that $\PGL_{m_e}$ only acts on the copy of $U_e \otimes U_e^* \subseteq W_{i_1} \otimes W_{i_2}$. In fact, because of the structure of $T(\Gamma,\bfm)$,  $\PGL_{m_e}$ only acts on the Kronecker factor $\bfu_e = \Id_{m_e}^{(i_1,i_2)} \otimes  \left( \bigotimes_{j \neq i_1,i_2} v_0^{(j)} \right) \in U_e \otimes U_e^* \otimes \bigotimes_{j \neq i_1,i_2} \bbC^1$. 

For $A \in \PGL_{m_e}$, we have $A \cdot \bfu_e = ( A^{-1}\Id_{m_e}^{(i_1,i_2)} A ) \otimes  \left( \bigotimes_{j \neq i_1,i_2} v_0^{(j)} \right) = \bfu_e$. Therefore $\PGL_{m_e}$ stabilizes $\bfu_e$.
\end{proof}

The image of the homomorphism $\Psi$ restricted to $\bigtimes_{e} \GL^\Delta_{m_e}$ is a subgroup
\[
\calG_{\Gamma,\bfm} \simeq \bigtimes_{e \in \bfe(\Gamma)}\PGL_{m_e} \subseteq G(W_k : k \in \bfv(\Gamma)) 
\]
called \emph{gauge subgroup} of $\Gamma$ with bond dimensions $\bfm$. A consequence of Lemma \ref{lemma: gauge stabilizes} is that the graph tensor $T(\Gamma,\bfm)$ is stabilized by $\calG_{\Gamma,\bfm}$. 

Denote by $\bfg_{\Gamma,\bfm}$ the Lie algebra of the gauge subgroup $\calG_{\Gamma,\bfm}$ of $\Gamma$.

\subsection{Isotropy group of graph tensors}\label{subsec: dim tns via isotropy}

The main result of this section is that the identity component of the isotropy group $G_{T(\Gamma,\bfm)}$ of a graph tensor coincides with the gauge subgroup. This generalizes the known results for the iterated matrix multiplication tensor \cite{dGro:VarsOptAlgIIsoGrps,Ges:Geometry_of_IMM}, that is the graph tensor associated to the cycle graph. We prove a more general form of this fact in Theorem \ref{thm: isotropy group when kron with star}; the result on graph tensors will follow via an inductive argument in Corollary \ref{corol: lie algebra for graph tensors}.

\begin{theorem}\label{thm: isotropy group when kron with star}
	Let $T' \in  \bbC^1 \otimes \bigotimes_{j =1}^d W'_j$ be a concise tensor of order $d+1$. Let $\Sigma = (\bfv(\Sigma), \bfe(\Sigma))$ be the graph on $d+1$ vertices $\bfv(\Sigma) = \{ 0 \vvirg d\}$ with edge set $\bfe(\Sigma) = \{ e_1 \vvirg e_k\}$, where $e_j = \{ 0,j\}$. Let $\bfm = (m_j : j =1 \vvirg k)$ be a set of bond dimensions on $\Sigma$. Let 
	\[
	S:=T(\Sigma,\bfm) \in  \bbC^{m_1 \cdots m_k} \otimes \bbC^{m_1} \ootimes \bbC^{m_k} \otimes \bbC^1 \ootimes \bbC^1
	\]
	be the associated graph tensor. Let $T = S \boxtimes T' \in V_0 \ootimes V_d$. Then 
	\[
	\frakg_{T} = \frakh_{T'} + \bfg_{\Sigma,\bfm} \subseteq \frakg(V_0 \vvirg V_d)
	\]
	where $\bfg_{\Sigma,\bfm}$ is the Lie algebra of the gauge subgroup $\calG_{\Sigma,\bfm}$ of $\Sigma$ and $\frakh_{T'}$ is the isotropy Lie algebra of $T'$ in $\frakg( \bbC^1, W'_1 \vvirg W'_d)$.
\end{theorem}
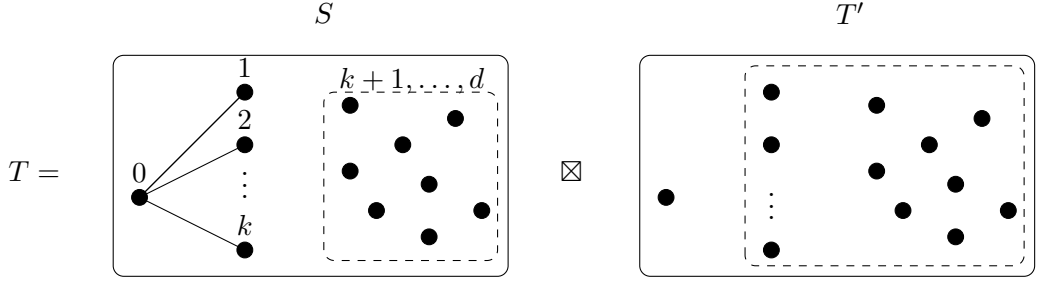
\begin{figure}
	\usetikzlibrary{arrows}
	\usetikzlibrary{matrix}
	\usetikzlibrary{positioning}
	\usetikzlibrary{snakes}
	\usetikzlibrary{calc}
	\begin{tikzpicture}[scale=.7]
	{
		\draw[fill=black] (0,0) circle (.15cm);
		\draw[fill=black] (2,2) circle (.15cm);
		\draw[fill=black] (2,1) circle (.15cm);
		\draw[fill=black] (2,-1) circle (.15cm);
		\draw (2,.4) node {$\vdots$};
		\draw (0,0)--(2,2);
		\draw (0,0)--(2,1);
		\draw (0,0)--(2,-1);
		\draw (0,0)--(2,2);
		\draw[] (0,0.1) node[above] {$0$};
		\draw[] (2,2.1) node[above] {$1$};
		\draw[] (2,1.1) node[above] {$2$};
		\draw[] (2,-0.9) node[above] {$k$};
		{
			\tikzset{shift={(-1,0)}}
			\draw[fill=black] (5,1.75) circle (.15cm);
			\draw[fill=black] (7,1.5) circle (.15cm);
			\draw[fill=black] (6,1) circle (.15cm);
			\draw[fill=black] (5,0.5) circle (.15cm);
			\draw[fill=black] (6.5,0.25) circle (.15cm);
			\draw[fill=black] (5.5,-0.25) circle (.15cm);
			\draw[fill=black] (7.5,-0.25) circle (.15cm);
			\draw[fill=black] (6.5,-0.75) circle (.15cm);
			\draw[rounded corners] (.5, -1.5) rectangle (8,2.7) {};
			\draw[rounded corners,dashed] (4.5, -1.2) rectangle (7.8,2) {};
			\draw[] (4.6,2.65) node[below right] {$k+1, \dots , d$};
		}
	}
	{
		\tikzset{shift={(10,0)}}
		\draw[fill=black] (0,0) circle (.15cm);
		\draw[fill=black] (2,2) circle (.15cm);
		\draw[fill=black] (2,1) circle (.15cm);
		\draw[fill=black] (2,-1) circle (.15cm);
		\draw (2,0) node {$\vdots$};
		{
			\tikzset{shift={(-1,0)}}
			\draw[fill=black] (5,1.75) circle (.15cm);
			\draw[fill=black] (7,1.5) circle (.15cm);
			\draw[fill=black] (6,1) circle (.15cm);
			\draw[fill=black] (5,0.5) circle (.15cm);
			\draw[fill=black] (6.5,0.25) circle (.15cm);
			\draw[fill=black] (5.5,-0.25) circle (.15cm);
			\draw[fill=black] (7.5,-0.25) circle (.15cm);
			\draw[fill=black] (6.5,-0.75) circle (.15cm);
			\draw[rounded corners] (.5, -1.5) rectangle (8,2.7) {};
			\draw[rounded corners,dashed] (2.5, -1.3) rectangle (7.8,2.5) {};
		}
	}
	\draw (3.5,3.5) node { $S$};
	\draw (13.5,3.5) node {$T'$};
	\draw (-2,.5) node { $T = $};
	\draw (8.2,.5) node { $\boxtimes$};
	\end{tikzpicture}
	\caption{The tensor $T$ in Theorem \ref{thm: isotropy group when kron with star}: the Kronecker product of a star tensor $S$ centered at vertex $0$ and a tensor $T'$ whose $0$-th factor is $1$-dimensional.}
\end{figure}

\begin{proof}
	The inclusion $\frakh_{T'} + \bfg_{\Sigma,\bfm} \subseteq \frakg_{T}$ is immediate. 
	
	For $j = 1 \vvirg k$, write $V_{j} = U_j \otimes W'_j$ where $U_j =\bbC^{m_j}$. Write $V_0 = \bbC^1 \otimes U_1^* \ootimes U_k^*$. For $j = 1\vvirg k$, let $\{ u^{j}_{i_j} : i_j = 1 \vvirg m_j\}$ be a basis of the $U_j$; let $\{ u_{i_1 \vvirg i_k}^{0} : i_j = 1 \vvirg m_j\}$ be the basis of $V_0 \simeq U_1^* \ootimes U_k^*$ dual to the induced basis $\{u^1_{i_1} \ootimes u^{k}_{i_k} : i_j = 1 \vvirg m_j\}$ of $U_1 \ootimes U_k$. Therefore 
	\[
	S = \sum_{i_1 \vvirg i_k} u^{(0)}_{i_1 \vvirg i_k} \otimes u^{(1)}_{i_1} \ootimes u^{(k)}_{i_k} \ \otimes u^{k+1}_0 \ootimes u^d_0 
	\]
	where for $j = k+1 \vvirg d$, $u^j_0$ is a generator of the corresponding $\bbC^1$ factor.
	
	Let $\bfX = (X_0 \vvirg X_d) \in \frakg(V_0 \vvirg V_d)$. Suppose $\bfX \in \frakg_T$, that is $\bfX.T = 0$. By Leibniz's rule $\bfX.T = \sum_{j=0} ^d X_j . T = 0$. 
	
	Write $X_0 = ({(x^0)}^{i'_1 \vvirg i'_k}_{i_1 \vvirg i_k})$ in the chosen basis: we have
	\[
	X_{0}.T = (X_{0} .S) \boxtimes T' = \left[\sum_{\substack{i_1 \vvirg i_k \\ i'_1 \vvirg i'_k}} {(x^{0})}^{i'_1 \vvirg i'_k}_{i_1 \vvirg i_k}u^{0}_{i'_1 \vvirg i'_k} \otimes  u^{1}_{i_1} \ootimes u^{k}_{i_k} \right] \boxtimes T'  .
	\]
	For $j = 1\vvirg k$, write $X_{j} \in \frakgl(V_{j})$ as $X_j = \sum \Delta_j^{(\rho)} \boxtimes \Theta_j^{(\rho)}$ where $\Delta_j^{(\rho)} = ((\delta^{\rho,j})^{i_j}_{i'_j}) \in \frakgl(U_j)$ and $ \Theta_j^{(\rho)} \in \frakgl(W'_j)$; then 
	\begin{align*}
		X_{j}.T &=  \sum_\rho \left[  \Delta_j^{(\rho)} .S \right] \boxtimes \left[ \Theta_j^{(\rho)} .T' \right] = \\ &=\sum_\rho \left[ \sum_{i_1 \vvirg i_k, i'_j} u^{0}_{i_1 \vvirg i_k} \otimes  u^{1}_{i_1} \ootimes {(\delta^{\rho,j})^{i_j}_{i'_j}}u^{j}_{i_j} \ootimes u^{k}_{i_k} \right] \boxtimes \left[\Theta_j^{(\rho)} . T'\right]. 
	\end{align*}
	If $j > k$, then $V_j = \bbC^1 \otimes W'_j$ and we have $X_{j}.T = S \boxtimes X_j.T'$.
	
	For indices $i_1^* \vvirg i_k^* , \ti_1 \vvirg \ti_k$, write $\bfX. T = u^0_{i_1^* \vvirg i_k^*} \otimes u^1_{\ti_1} \ootimes u^k_{\ti_k} \boxtimes T^{i_1^* \vvirg i_k^*}_{ \ti_1 \vvirg \ti_k}$ for tensors $T^{i_1^* \vvirg i_k^*}_{ \ti_1 \vvirg \ti_k} \in W'_1 \ootimes W'_d$. Since $u^0_{i_1^* \vvirg i_k^*} \otimes u^1_{\ti_1} \ootimes u^k_{\ti_k}$ are linearly independent, the condition $\bfX.T = 0$ is equivalent to $T^{i_1^* \vvirg i_k^*}_{ \ti_1 \vvirg \ti_k} = 0$ for every $i_1^* \vvirg i_k^* , \ti_1 \vvirg \ti_k$.
	
	Note that if $(i_1^* \vvirg i_k^*)$ and $( \ti_1 \vvirg \ti_k)$ differ in at least two entries, then $T^{i_1^* \vvirg i_k^*}_{ \ti_1 \vvirg \ti_k}$ only depends on $X_0$: indeed, the summands $X_j.T$ for $j \neq 0$ only give rise to terms where $(i_1^* \vvirg i_k^*)$ and $( \ti_1 \vvirg \ti_k)$ differ in at most one entry. Write $X_0 = X_0' + X_0''$ where $X_0'$ is the component where $(i_1^* \vvirg i_k^*)$ and $( \ti_1 \vvirg \ti_k)$ differ in at least two entries and $X_0''$ is the complementary component. In particular, $X_0'$ is the only component of $\bfX$ which contributes to $T^{i_1^* \vvirg i_k^*}_{ \ti_1 \vvirg \ti_k}$ when the two sets of indices differ in at least two entries. By linearity, the discussion above shows $X_0'.T =0$. Since $T$ is concise, Lemma \ref{lemma: partial lie algebra} implies that $X_0' = 0$. This shows 
	\[
	X_{0} = Y_1 \otimes \id_{U_2^* \ootimes U_k^*} + \cdots + \id_{U_1^* \ootimes U_{k-1}^*} \otimes Y_k
	\]
	with $Y_j \in \frakgl(U_j^*)$. Hence, we may renormalize $\bfX$ using $\bfg_{\Gamma,\bfm}$ and obtain $X_{0} = 0$. In particular, we reduced the analysis to $\bfX \in \frakg( V_j : j \neq 0)$.
	
	Consider $\bfX \in \frakg_T \cap \frakg( V_j : j \neq 0)$. By Lemma \ref{lemma: partial lie algebra}, we have 
	\begin{equation}\label{eqn: intersection all but V0}
	\frakg_T \cap \frakg( V_j : j \neq 0) = \bigcap_{R \in \Im (\Flat(T))}  \frakg_R ,
	\end{equation}
	where $\Flat(T) : V_0^* \to V_1 \ootimes V_d$ is the $0$-th flattening map. For indices $(i_1 \vvirg i_k)$, write $T'(i_1 \vvirg i_k) = \Flat(T) ( u_{i_1 \vvirg i_k}^{(0)})=u^{(1)} _{i_1} \ootimes u^{(k)}_{i_k} \boxtimes T' $. The intersection in \eqref{eqn: intersection all but V0} can be reduced to a set of generators of $\Im(\Flat(T))$; therefore we obtain
	\[
	\frakg_T \cap \frakg( V_j : j \neq 0)  = \bigcap_{i_1 \vvirg i_k} \frakg_{T'(i_1 \vvirg i_k)}.
	\]
	Since $T'(i_1 \vvirg i_k)$ is not concise in $V_1 \ootimes V_d$, we have $\frakg_{T'(i_1 \vvirg i_k)} = \frakh_{T'(i_1 \vvirg i_k)} \oplus \frakp_{i_1 \vvirg i_k}$, where $\frakh_{T'(i_1 \vvirg i_k)}$ is the annihilator of $T'(i_1 \vvirg i_k)$ in $\frakgl(\langle u^{1} _{i_1} \rangle \otimes W'_1) \ooplus \frakgl(\langle u^{k} _{i_k} \rangle \otimes W'_k) \oplus \frakgl(V_{k+1}) \ooplus \frakgl(V_d)$ and $\frakp_{i_1 \vvirg i_k}$ is the parabolic subspace which annihilates $(u^{1}_{i_1}\otimes W'_1) \ootimes (u^{k} _{i_k} \otimes W'_k) \otimes V_{k+1} \ootimes V_d$, that is 
	\begin{equation*}\label{ort}
		\begin{aligned}
			\frakp_{i_1 \vvirg i_k} =& \left[(\langle u^{1} _{i_1} \rangle^\perp \otimes {W'_1}^{*} )\otimes (U_1 \otimes W'_1)\right] \ooplus \left[ ( \langle u^{k} _{i_k} \rangle^\perp \otimes {W'_k}^{*} )\otimes (U_k \otimes W'_k)\right].
		\end{aligned}
	\end{equation*}
	Since $T'(i_1 \vvirg i_k) =u^{1} _{i_1} \ootimes u^{k}_{i_k} \boxtimes T'$, we have
	\[
	\frakh_{T'(i_1 \vvirg i_k) } = \Id_{\langle u^{1}_{i_1}\rangle \ootimes \langle u^{k}_{i_k}\rangle} \otimes \frakg_{T'}, 
	\]
	regarded as a subalgebra acting on the subspace $u^{1}_{i_1} \ootimes u^{k}_{i_k} \otimes W'_1 \ootimes W'_k \otimes V_{k+1} \ootimes V_d$. 
	
	Observe that, as a subspace of $\End(V_1 \ootimes V_d)$, we have 
	\[
	\frakg_{T'(i_1 \vvirg i_k)} = \left[ \Id_{\langle u^{1}_{i_1}\rangle \ootimes \langle u^{k}_{i_k}\rangle} \otimes \frakg_{T'} \right] \oplus \frakp_{i_1 \vvirg i_k} = \Bigl[ \Id_{U_1 \ootimes U_k} \otimes \frakg_{T'} \Bigr] \oplus \frakp_{i_1 \vvirg i_k}.
	\]
	This follows directly from Leibniz rule and the fact that, for every $i_1 \vvirg i_k$, $\Id_{U_1\ootimes U_k} = \Id_{ \langle u^1_{i_1} \ootimes u^k_{i_k} \rangle} + P_{i_1 \vvirg i_k}$ where $P_{i_1 \vvirg i_k} \in \frakp_{i_1 \vvirg i_k}$. We deduce
	\[
	\frakg_T \cap \frakg( V_j : j \neq 0) = \bigcap_{i_1 \vvirg i_k} \left[ \left( \Id_{U_1 \ootimes U_k} \otimes \frakg_{T'} \right) \oplus \frakp_{i_1 \vvirg i_k} \right]
	\]
	and by Lemma \ref{lemma: intersection vs direct sum}, we have $ \frakg_T \cap \frakg( V_j : j \neq 0)   = \left( \Id_{U_1 \ootimes U_k} \otimes \frakg_{T'}\right)  \oplus \bigcap_{i_1 \vvirg i_k} \frakp_{i_1 \vvirg i_k} = \Id_{U_1 \ootimes U_k} \otimes \frakg_{T'}$ because $\bigcap_{i_1 \vvirg i_k} \frakp_{i_1 \vvirg i_k}  = 0$.
	
	This concludes the proof, as we showed 
	\[
	\frakg_T = \frakg_T + \bfg_{\Sigma,\bfm} = \frakg_T \cap \frakg( V_j : j \neq 0)  + \bfg_{\Sigma,\bfm} = \frakg_{T'} + \bfg_{\Sigma,\bfm}. \qedhere 
	\]
\end{proof}

Applying Theorem \ref{thm: isotropy group when kron with star} to graph tensors, we deduce the following result:
\begin{corollary}\label{corol: lie algebra for graph tensors}
	Let $\Gamma = (\bfv(\Gamma), \bfe(\Gamma))$ be a graph with $d$ vertices and let $\bfm = (m_e : e \in \bfe(\Gamma))$ be a set of bond dimensions on $\Gamma$. Let $T:=T(\Gamma,\bfm) \in \bigotimes_{j =1}^d W_j$ be the associated graph tensor. Then the isotropy Lie algebra of $T$ coincides with Lie algebra of the gauge subgroup of $\Gamma$; in symbols
	\[
	\frakg_T = \bfg_{\Gamma,\bfm}.
	\]
\end{corollary}
\begin{proof}
	We proceed by induction on the number of vertices $d$. If $d= 1$, the statement is clear as $T$ is a single vector, with trivial isotropy Lie algebra.
	
	Suppose $\Gamma$ is a graph with $d+1$ vertices and write $\bfv(\Gamma) = \{ 0 \vvirg d\}$. Let $\Sigma$ be the subgraph of $\Gamma$ given by the edges incident to the vertex $0$. In other words $\bfv(\Sigma) = \{ 0 \vvirg d\}$, $\bfe(\Sigma) = \{ e \in \bfe(\Gamma): 0 \in e\}$. Let $\Gamma'$ be the graph with $\bfv(\Gamma') = \{ 0 \vvirg d\}$ and $\bfe(\Gamma') = \bfe(\Gamma) \setminus \bfe(\Sigma)$ and let $\bfm'' , \bfm'$ be the corresponding subsets of the collection of bond dimensions $\bfm$. Write $S = T(\Sigma,\bfm'')$ and $T' = T(\Gamma',\bfm')$; then
	\[
	T = S \boxtimes T'.
	\]
	By the induction hypothesis, $\frakg_{T'} = \bfg_{\Gamma', \bfm'}$ and $\frakg_S = \bfg_{\Sigma,\bfm''}$. By Theorem \ref{thm: isotropy group when kron with star}
	\[
	\frakg_T = \frakg_{T'} + \bfg_{\Sigma,\bfm''} = \bfg_{\Gamma', \bfm'} + \bfg_{\Sigma, \bfm''} = \bfg_{\Gamma,\bfm},
	\]
	and this concludes the proof.
\end{proof}

\subsection{Additional results on isotropy groups}

In this section, we prove a generalization of \cite[Thm. 4.1(iii)]{ConGesLanVenWan:GeometryStrassenAsyRankConj}; it does not have a direct application in this work but it is of interest on its own right.

Given two spaces $V,W$, there is a natural embedding $\GL(V) \to \GL(V \otimes W)$ defined by $g \mapsto g \otimes \Id_W$; correspondingly the Lie algebra $\frakgl(V)$ can be regarded as a subalgebra of $\frakgl(V \otimes W)$. In particular, if $\frakg \subseteq \frakgl(V)$ is a subalgebra, then $\frakg$ is naturally identified with a subalgebra of $\frakgl(V \otimes W)$. 

\begin{proposition}\label{prop: stable lie algebra when tensoring with trivial}
	Let $T \in V_1 \ootimes V_d$ and $S \in W_1 \ootimes W_d$ be concise tensors. Assume $\frakg_T  = \{0\} \subseteq \frakg(V_1 \vvirg V_d)$ . Then
	\[
	\frakg_{T \boxtimes S} =   \frakg_S 
	\]
	regarded as a subalgebra of $\frakg(V_1 \otimes W_1 \vvirg V_d \otimes W_d)$.
\end{proposition}
\begin{proof}
	The inclusion 
	\[
	\frakg_S \subseteq \frakg_{T \boxtimes S}
	\]
	is immediate from the definition of Kronecker product.
	
	Let $\bfX \in   \frakg_{T \boxtimes S}$. Write $\bfX = (X_1 \vvirg X_d)$ with $X_k \in \frakgl(V_k \otimes W_k)$. Our goal is to show that $X_k = \Id_{V_k} \otimes Z_k$ for some $Z_k \in \frakgl(W_k)$ with $\bfZ := (Z_1 \vvirg Z_d) \in \frakg_S$. 
	
	For every $p = 1 \vvirg d$, fix bases $\{ v^{p}_j : j = 1 \vvirg \dim V_p\}$ of $V_p$ and similarly for $W_p$. Write
	\begin{align*}
		T &= \sum T^{i_1 \vvirg i_d} v^{1}_{i_1} \ootimes v^{d}_{i_d}, \\
		S &= \sum S^{j_1 \vvirg j_d} w^{1}_{j_1} \ootimes w^{d}_{j_d}. 
	\end{align*}
	For $k = 1 \vvirg d$, write $(x_k)^{ij}_{i'j'}$ for the entries of $X_k$ with respect to the basis $v^{k}_i \otimes w^k_j$. By Leibniz's rule, the condition $\bfX. (T\boxtimes S) = 0$ is equivalent to
	\begin{equation}\label{eqn: linear system for T kron S}
	\sum_{k=1}^d ({x_k})^{i_kj_k}_{i_k'j_k'} T^{i_1 \vvirg i_k' \vvirg i_d}S^{j_1\vvirg j_k' \vvirg j_d} = 0 \quad \text{for every $i_1 \vvirg i_d, j_1 \vvirg j_d$},
	\end{equation}
	where we use the summation convention that repeated upper and lower indices are to be summed over their range.
	
	For every $j_1 \vvirg j_d$, and for every $k = 1 \vvirg d$, define $Y_k (j_1 \vvirg j_d) \in \frakgl(V_k)$ by
	\[
	(y_k (j_1 \vvirg j_d))^{i_k}_{i_k'} =  ({x_k})^{i_kj_k}_{i_k'j_k'} S^{j_1 \vvirg j_k' \vvirg j_d}.
	\]
	Regard $\bfY(j_1 \vvirg j_d) = (Y_1 (j_1 \vvirg j_d) \vvirg Y_d(j_1 \vvirg j_d))$ as an element of $\frakg(V_1 \vvirg V_d)$. From \eqref{eqn: linear system for T kron S}, we deduce that $\bfY(j_1 \vvirg j_d)$ satisfies $\bfY.T = 0$ and therefore $\bfY \in \frakg_T$. From the hypothesis $\frakg_T = \{0\}$ and therefore, for every $k$, there exists $\lambda_k(j_1 \vvirg j_d)$ such that $Y_k(j_1\vvirg j_d) = \lambda_k (j_1 \vvirg j_d) \Id_{V_k}$  and $\sum_k \lambda_k(j_1 \vvirg j_d) = 0$.
	
	Since $Y_k(j_1\vvirg j_d)$ is a multiple of the identity, we have
	\begin{equation*}\label{eqn: conditions on X for Y identity}
		\begin{aligned}
			0 &= ( y_k(j_1\vvirg j_d) )^{i_k}_{i_k'} = (x_k)^{i_kj_k}_{i_k'j_k'} S^{j_1 \vvirg j_k' \vvirg j_d} \quad \text{for $i_k \neq i_k'$},\\
			0 &= ( y_k(j_1\vvirg j_d) )^{i_k}_{i_k} - ( y_k(j_1\vvirg j_d) )^{1}_{1} = [ (x_k)^{i_kj_k}_{i_kj_k'} - (x_k)^{1j_k}_{1j_k'}]  S^{j_1 \vvirg j_k' \vvirg j_d} .
		\end{aligned}
	\end{equation*}
	In other words, if $i_k \neq i_k'$, setting $Z_k (i_k,i'_k) \in \frakgl(W_k)$ to be defined by $(z_k(i_k,i'_k) )^{j_k}_{j_k'} = (x_k)^{i_kj_k}_{i_k'j_k'}$, we have $Z_k (i_k,i'_k). S = 0$. This means that $Z_k (i_k,i'_k) \in \frakg_S \cap \frakgl(W_k)$: since $S$ is concise, Lemma \ref{lemma: partial lie algebra} implies $Z_k(i_k,i'_k) = 0$. This shows that $ (x_k)^{i_kj_k}_{i_k'j_k'} = 0$ whenever $i_k \neq i_k'$. Similarly, if $i_k \geq 2$, setting $(z_k(i_k))^{j_k}_{j_k'} = (x_k)^{i_kj_k}_{i_kj_k'} - (x_k)^{1j_k}_{1j_k'}$, we have $Z_k(i_k).S = 0$, hence $Z_k(i_k) = 0$ and therefore $(x_k)^{i_kj_k}_{i_kj_k'} =  (x_k)^{1j_k}_{1j_k'}$ for every $i_k$.
	
	We deduce that $X_k = \Id_{V_k} \otimes Z_k$ for some $Z_k \in \frakgl(W_k)$. Now, let $\bfZ = (Z_1 \vvirg Z_k)$. We conclude
	\[
	0= \bfX.(T \boxtimes S) = \bfZ.(T\boxtimes S) = T \boxtimes \bfZ.S 
	\]
	and therefore $\bfZ \in \frakg_S$. This concludes the proof.
\end{proof}

\section{Dimension of Tensor Network varieties} \label{section: dimension}

We provide an upper bound on $\dim \calTNS^\Gamma_{\bfm,\bfn}$ for every $\bfm$ and $\bfn$. 

First, we give a definition following \cite{LanQiYe:GeomTensorNetwork}.
\begin{definition}
 Let $(\Gamma,\bfm,\bfn)$ be a tensor network. A vertex $v \in \bfv$ is called 
 \begin{itemize}
  \dotitem \emph{subcritical} if $\prod_{e \ni v} m_e \geq n_v$; \emph{strictly subcritical} if the inequality is strict;
  \dotitem \emph{supercritical} if $\prod_{e \ni v} m_e \leq n_v$; \emph{strictly supercritical} if the inequality is strict;
  \dotitem \emph{critical} if $v$ is both subcritical and supercritical.
 \end{itemize}
The tensor network $(\Gamma,\bfm,\bfn)$ is called [strictly] subcritical (resp. supercritical) if all its vertices are [strictly] subcritical (resp. supercritical).
\end{definition}

First, we determine a reduction which allows us to assume that the bond dimensions associated to the edges incident to a fixed vertex are \emph{balanced}, in a way made precise in Lemma \ref{lemma: superfluous bond}.

Then, we provide a second reduction, proving that the tensor network variety of a tensor network having strictly supercritical vertices can be realized via a vector bundle construction as a natural extension of the tensor network variety where the strictly supercritical vertices are reduced to be critical.

Finally, we prove an upper bound for $\dim \calTNS_{\bfm , \bfn}^{\Gamma}$ in the subcritical range. This upper bound reduces to an equality in the critical case. 

Recall that from \eqref{eqn: dimension of fibers}, we have 
\[
 \dim \calTNS^\Gamma_{\bfm,\bfn} = \dim \Hom( W_1 \vvirg W_d, V_1 \vvirg V_d) - \dim \Phi^{-1}(T).
\]

\subsection{Reduction of bond dimension}\label{subsec: reduction bond dimension}

We already observed in Remark \ref{rmk: dimension 1} that we may always assume bond dimensions are at least $2$. Here, we show that if they are ``too unbalanced'', then they can be reduced without affecting the dimension of the tensor network variety. 

We say that a tensor network $(\Gamma,\bfm,\bfn)$ has \emph{overabundant bond dimension} if there exist a vertex $v \in \bfv(\Gamma)$ and an edge $e \in \bfe(\Gamma)$ incident to $v$ such that 
\begin{equation}\label{SuperflEq}
 m_{e} > n_v \prod_{e' \ni v, e' \neq e}m_{e'}.
\end{equation}
The following result shows that overabundant bond dimensions do not contribute to the dimension of the tensor network variety.
\begin{lemma}\label{lemma: superfluous bond}
Let $(\Gamma, \bfm,\bfn)$ be a tensor network. Fix $v \in \bfv(\Gamma)$, let $k$ be the degree of the vertex $v$ and $\{e_1 \vvirg e_{k}\}$ be the edges incident to $v$; assume $m_{e_1} \leq \cdots \leq m_{e_k}$. If \eqref{SuperflEq} holds for $v$, that is 
\[
 m_{e_k} > n_v \cdot m_{e_1} \cdots m_{e_{k-1}},
\]
then 
\[
 \calTNS_{\bfm , \bfn} = \calTNS_{\bar{\bfm},\bfn}
\]
where $\bar{\bfm}$ is defined by $\bar{m}_e = m_e$ if $e \neq e_k$ and $\bar{m}_{e_k} = n_v\cdot m_1 \cdots m_{e_{k-1}} $.
\end{lemma}
\begin{proof}
 Let $T \in \calTNS_{\bfm , \bfn} \subseteq V_1 \ootimes V_d$ be a generic element and let $(X_1 \vvirg X_d) \in \Hom (W_1 \vvirg W_d ; V_1 \vvirg V_d)$ be an element such that $(X_1 \vvirg X_d) \cdot T(\Gamma,\bfm) = T$.
 
 Suppose $v = d$ and $e_j = \{d ,j\}$ for $ j =1 \vvirg k$. Write $U_j = \bbC^{m_j}$; let $W_d = U_1^* \ootimes U_k^*$, so that, for $j = 1 \vvirg k$, we have $W_j = U_j \otimes W_j'$ where $W_j'$ depends on the other edges incident to the vertex $j$. 
 
 Regard $X_d$ as a tensor in $W_d^* \otimes V_d = U_1 \ootimes U_k \otimes V_d$. Since $ m_{e_k} > m_{e_1} \cdots m_{e_{k-1}} \cdot n_v$, $X_d$ is not concise on the factor $U_k$: let $\bar{U}_k \subseteq U_k$ with $\dim \bar{U}_k =  m_{e_1} \cdots m_{e_{k-1}} \cdot n_v$ be a subspace such that $X_d \in U_1 \ootimes U_{k-1} \otimes \bar{U}_k \otimes V_d$. Correspondingly, let $\bar{U}_{k}^* = U_k^* / \bar{U}_{k}^\perp$. Note that $T(\Gamma,\bar{\bfm})$ coincides with the image of $T(\Gamma,\bfm)$ via the projection $U_k^* \to \bar{U}_k^*$ on the $d$-th factor.

 Now, define $\bar{W}_d = U_1^* \ootimes U_{k-1}^* \otimes \bar{U}_k^*$ and $\bar{W}_k = W_k' \otimes \bar{U}_k$. Let $\bar{X}_d = X_d$ be the linear map regarded as an element of $\Hom(\bar{W}_d, V_d)$. Moreover, the space $\Hom (W_k , V_k) = (W_k' \otimes U_k)^* \otimes V_k= {W_k'}^* \otimes U_k^* \otimes V_k$ naturally projects onto ${W_k'}^* \otimes \bar{U}_k^* \otimes V_k = \Hom(\bar{W}_k , V_k)$: let $\bar{X}_k$ be the image of $X_k$ under this projection. 

 Now, one can verify that
 \[
T = (X_1 \ootimes X_d) \cdot  T(\Gamma, \bfm) = (\bar{X}_1 \ootimes \bar{X}_d ) \cdot T(\Gamma,\bar{\bfm})
 \]
 where $\bar{X}_v = X_v$ if $v \neq k,d$.
\end{proof}

\subsection{Reduction for supercritical vertices}

The reduction of this section appeared already in \cite{LanQiYe:GeomTensorNetwork}. We include it here for completeness.

For a vector space $V$ with $\dim V = n$ and an integer $k \leq n$, let $\bfG(k,V)$ be the Grassmannian of $k$-dimensional linear subspaces of $V$. Recall that $\dim \bfG(k,V) = k(n-k)$. The variety $\bfG(k,V)$ has a \emph{tautological bundle} 
\[
\sigma : \calS \to \bfG(k,V);
\]
the fiber of $\calS$ over a point $[E] \in \bfG(k,V)$ is the plane $E$ itself: $\calS_{[E]} = E$. 

\begin{proposition}\label{prop: reduction supercritical}
Let $(\Gamma,\bfm,\bfn)$ be a tensor network. Suppose that the vertex $d \in \bfv(\Gamma)$ is supercritical and write $N = \dim W_d = \prod_{e \ni d} m_e$. Let $\bfn' = (n'_v : v \in \bfv(\Gamma))$ be the $d$-tuple of local dimensions defined by $n'_v = n_v$ if $v \neq d$ and $n'_d = N$.

Then
\[
 \dim \calTNS_{\bfm,\bfn}^\Gamma =  N(n_d-N) + \dim \calTNS_{\bfm,\bfn'}^\Gamma .
\]
\end{proposition}
\begin{proof}
 Let $\calS_d^{V_1 \ootimes V_{d-1}}$ be the vector bundle over the Grassmannian $\bfG(N,V_d)$ whose fiber over a plane $[E]$ is $V_1 \ootimes V_{d-1} \otimes E$; this is the tautological bundle augmented by the trivial bundle with constant fiber $V_1 \ootimes V_{d-1}$. Consider the diagram
 \[
  \xymatrix{ 
  &\calS_d^{V_1 \ootimes V_{d-1}} \ar[dr]^{\pi} \ar[dl]_\sigma & \\
  \bfG(N,V_d) & & V_1 \ootimes V_d
 }\]
 where the second projection $\pi: ([E], T) \mapsto T$ maps an element of the bundle to its fiber component. By conciseness, this projection is generically one-to-one. 
 
Consider the subbundle of $\calS_d^{V_1\ootimes V_{d-1}}$ whose fiber at $[E]$ is $\calTNS^\Gamma_{\bfm,\bfn'}$ where the $d$-th factor is identified with $E$. Let $\underline{\calTNS}^\Gamma_{\bfm, \calS}$ be the total space of this subbundle. We have 
\[
 \dim \underline{\calTNS}^\Gamma_{\bfm, \calS} = \dim \bfG(N,V_d) + \dim \calTNS^{\Gamma}_{\bfm,\bfn'} = N(n_d-N) + \dim \calTNS_{\bfm,\bfn'}^\Gamma .
\]
The projection $\pi$ is generically one-to-one and maps $\underline{\calTNS}^\Gamma_{\bfm, \calS}$ surjectively onto $\calTNS^\Gamma_{\bfm,\bfn}$. Therefore $\dim \calTNS^{\Gamma}_{\bfm,\bfn} = \dim \underline{\calTNS}^\Gamma_{\bfm, \calS}$ and this concludes the proof.
\end{proof}

Iteratively applying Proposition \ref{prop: reduction supercritical}, one can reduce all strictly supercritical vertices to critical vertices.

\begin{theorem}\label{thm: reduction supercritical}
 Let $(\Gamma,\bfm,\bfn)$ be a tensor network. For every $v \in \bfv(\Gamma)$ let $N_v = \prod_{e \ni v} m_e$. Let $\bfn'$ be the set of local dimensions defined by $n_v' = \min\{ N_v, n_v\}$. Then
 \[
  \dim \calTNS^{\Gamma}_{\bfm,\bfn} = \sum_{v\in\bfv(\Gamma)} n_v' (n_v-n_v') + \dim \calTNS^\Gamma_{\bfm,\bfn'}.
 \]
\end{theorem}

Note that the tensor network $(\Gamma, \bfm,\bfn')$ appearing in Theorem \ref{thm: reduction supercritical} is, by definition, subcritical.

It remains to understand $\dim \calTNS^{\Gamma}_{\bfm,\bfn}$ in the subcritical range.

\subsection{Subcritical range}

We will provide an upper bound for $\dim \calTNS_{\bfm,\bfn}^\Gamma$ when the tensor network $(\Gamma,\bfm,\bfn)$ is subcritical. The upper bound is obtained, following \eqref{eqn: dimension of fibers}, by determining a lower bound on the dimension of the generic fiber of the parametrization of $\calTNS_{\bfm,\bfn}^\Gamma$.

\begin{theorem}\label{thm: fiber and dim gauge-orbit}
 Let $(\Gamma,\bfm,\bfn)$ be a subcritical tensor network. Then the dimension of the generic fiber of the map $\Phi$ is bounded from below by the dimension of the $\calG_{\Gamma,\bfm}$-orbit of a generic element of $\Hom(W_1,\dots, W_d; V_1,\dots,V_d).$
\end{theorem}
\begin{proof}
	Let $T = (X_1 \ootimes X_d)\cdot T(\Gamma,\bfm)$, with $X_1 \ootimes X_d \in \Hom(W_1,\dots, W_d; V_1,\dots,V_d) $ a generic element. The fiber of $\Phi : \Hom(W_1 \vvirg W_d; V_1 \ootimes V_d) \to \calTNS^\Gamma_{\bfm,\bfn}$ over $T$ is 
	\begin{equation*}
	\Phi^{-1} (T) = \{Y_1\ootimes Y_d \in \Hom(W_1,\dots, W_d; V_1,\dots,V_d) : (Y_1\ootimes Y_d)\cdot T(\Gamma,\bfm) = T\}
	\end{equation*} 
	Since every vertex is subcritical, for every $j$, a generic element of $\Hom(W_j,V_j)$ is surjective. Let $Y_1 \ootimes Y_d \in \Phi^{-1}(T)$. By conciseness, $Y_j$ has the same image as $X_j$, therefore $Y_j$ is surjective as well, and there exists $g \in \GL(W_j)$ such that $Y_j = X_j g_j$. 

	For $X = X_1 \ootimes X_d$, and $g=g_1\ootimes g_d \in G(W_1, \dots, W_d)$, write $g.X= X_1g_1 \ootimes X_dg_d $. In particular, if $g \in \calG_{\Gamma,\bfm}$ then
	\begin{equation*}
	Y\cdot T(\Gamma,\bfm)= (g.X)\cdot T(\Gamma,\bfm)=(X_1 \ootimes X_d)(g_1\ootimes g_d)\cdot T(\Gamma,\bfm) = 	X \cdot T(\Gamma,\bfm)=T,
	\end{equation*}
	and the dimension of the fiber is bounded by
	\begin{equation}\label{dim fibre}
	\begin{aligned}
		\dim \Phi^{-1} (T) & = \dim \{Y: Y\cdot T(\Gamma,\bfm) = T \} \\
		&= \dim \{ g. X :g \in G(W_1, \dots, W_d), (g.X)\cdot T(\Gamma,\bfm) = T \} \\
		& \geq \dim \{g.X :g \in \calG_{\Gamma,\bfm}\} \\
		& = \dim (\calG_{\Gamma,\bfm}\cdot X). 
		\end{aligned}
	\end{equation}
	Therefore the dimension of the generic fiber is bounded from below by the dimension of the $\calG_{\Gamma,\bfm}$-orbit of a generic element of $\Hom(W_1,\dots, W_d; V_1,\dots,V_d)$. 
\end{proof}

Applying the Theorem of the Dimension of the Fibers \cite[Thm. 1.25]{Shaf:BasicAlgGeom1} to the $\calG_{\Gamma,\bfm}$-orbit of a generic element $X \in \Hom(W_1,\dots, W_d; V_1,\dots,V_d)$, we deduce the following corollary, which completes the proof of Theorem \ref{theorem: main}.
\begin{corollary}\label{corol: upper bound with Stab}
 Let $(\Gamma,\bfm,\bfn)$ be a subcritical tensor network with no overabundant bond dimension. Then 
 \[
 \dim \calTNS^{\Gamma}_{\bfm,\bfn} \leq \left [ \textsum_{v \in \bfv(\Gamma)} N_v n_v -d +1 \right] - \textsum_{e \in \bfe(\Gamma)} (m_e^2 -1) + \dim \Stab_{\calG_{\Gamma,\bfm}} (X)
 \]
where $N_v = \prod_{e \ni v} m_e$ and $X = X_1 \ootimes X_d$ with $X_v \in \Hom(W_v , V_v)$ generic.
\end{corollary}
\begin{proof}
From \eqref{eqn: dimension of fibers} $\dim \calTNS^\Gamma_{\bfm,\bfn} = \dim \Hom( W_1 \vvirg W_d, V_1 \vvirg V_d) - \dim \Phi^{-1}(T)$ where $T$ is a generic element of $\calTNS^\Gamma_{\bfm,\bfn}$.

Now $\dim \Hom (W_1, \ldots, W_d;V_1, \ldots , V_d)=\textsum_{v \in \bfv(\Gamma)} N_v n_v -d +1$. By Theorem \ref{thm: fiber and dim gauge-orbit}, 
\begin{align*}
\dim\Phi^{-1}(T)\geq &\dim \calG_{\Gamma,\bfm} \cdot X  = \\ 
&\dim \calG_{\Gamma,\bfm} - \dim \Stab_{\calG_{\Gamma,\bfm}}(X) = \\ 
& \textsum_{e \in \bfe(\Gamma)} (m_e^2 -1) - \dim \Stab_{\calG_{\Gamma,\bfm}} (X) ,
\end{align*}
where $X \in \Hom(W_1 \vvirg W_d, V_1 \vvirg V_d)$ is generic.
\end{proof}

\subsection{Sharpening the upper bound}\label{subsec: sharpening}

In this section, we study the upper bound obtained in Corollary \ref{corol: upper bound with Stab} and we provide sufficient conditions to have $\dim \Stab_{\calG_{\Gamma,\bfm}}(X) = 0$. These results will lead us to a complete proof of Corollaries \ref{cor: MPS} and \ref{cor: PEPS}.

\begin{definition}
Let $G$ be an algebraic group acting on an algebraic variety $V$. We say that the action is \emph{generically stable} if there exists an element $v \in V$ such that the stabilizer $\Stab_G(v)$ is a finite group.  
\end{definition}

In particular, the condition that the action of $\calG_{\Gamma,\bfm}$ on $\Hom(W_1 \vvirg W_d; V_1 \vvirg V_d)$ is generically stable is equivalent to the fact that the term $\dim \Stab_{\calG_{\Gamma,\bfm}}(X)$ in Corollary \ref{corol: upper bound with Stab} is $0$.

A rich theory has been developed in the study of stable group actions (and more generally semistable actions, which are beyond the scope of this paper) starting from \cite{KemNess:LengthVectorRepSpaces} and related works. We refer to \cite{MuFoKi:GeometricInvariantTheory} for the theory.

First, we provide a result on the cycle graph $C_d$, which yields the result on matrix product states in Corollary \ref{cor: MPS}.
\begin{proposition}\label{prop: MPS have stable action}
Let $(C_d, \bfm ,\bfn)$ be the tensor network on the cycle graph with constant bond dimension $\bfm = (m \vvirg m)$. Assume $n_j \geq 2$ for at least one index. Then the action of $\calG_{C_d,\bfm}$ on $\Hom(W_1 \vvirg W_d;V_1 \vvirg V_d)$ is generically stable. 
\end{proposition}
\begin{proof}
 Let $X_1 \ootimes X_d \in \Hom(W_1 \vvirg W_d; V_1 \vvirg V_d)$ be a generic element. Write $W_j = U_j \otimes U_{j+1}^*$ with $U_j = U_{j+1} = \bbC^m$. Then $X_j$ is a generic element of $U_j^* \otimes U_{j+1} \otimes V_j$, with $\dim V_j = n_j \geq 1$. For every $j$, write $X_j = \sum _{p = 1}^{n_j} X_j ^{(p)} \otimes v_p$ where $v_1 \vvirg v_{n_j}$ is a basis of $V_j$ and $X_j ^{(p)} \in U_j^* \otimes U_{j+1}$.

By genericity $X_j^{(1)}$ is a fixed isomorphism $X_j^{(1)} : U_j \to U_{j+1}$; after choosing bases in $U_j$, we write $X_j^{(1)} = \Id_{\bbC^m}$ in coordinates for $j = 1\vvirg d-1$ and $X_d^{(1)} : U_d \to U_1$ is a generic diagonal matrix. 

The stabilizer $\Stab_{\calG_{C_d,\bfm}}(X)$ is contained in the stabilizer of $X_1^{(1)} \ootimes X_d^{(1)}$: this is the centralizer of $X_d^{(1)}$; in coordinates this is the maximal torus $\Theta _m \subseteq \PGL_m$ of diagonal matrices in $\PGL_m^\Delta$, where $\PGL_m^\Delta \subseteq \calG_{C_d , \bfm} = \bigtimes_{j=1}^d \PGL(U_j)$ lies on the diagonal of the direct factors. Therefore $\Stab_{\calG_{C_d,\bfm}}(X) \subseteq \Theta_m$.

Now, there exists at least one index $j$ such that $n_j \geq 2$. Correspondingly, there is a map $X_j^{(2)} : U_j \to U_{j+1}$. Therefore $\Stab_{\calG_{C_d,\bfm}} (X) \subseteq \Stab_{\Theta_m} (X_j^{(2)})$. By genericity, $X_j^{(2)}$ has full rank and is not diagonal in the fixed basis, hence $\Stab_{\Theta_m} (X_j^{(2)})$ is trivial.

This shows that a generic $X \in \Hom(W_1 \vvirg W_d; V_1 \vvirg V_d)$ satisfies $\dim \Stab_{\calG_{C_d,\bfm}}(X) = 0$, hence the action of $\calG_{C_d,\bfm}$ on $\Hom(W_1 \vvirg W_d; V_1 \vvirg V_d)$ is generically stable.
\end{proof}

We generally believe that $\dim \Stab_{\calG_{\Gamma,\bfm}}(X) = 0$ in ``most cases'' at least when the bond dimensions are balanced. However, we cannot extend the argument of Proposition \ref{prop: MPS have stable action} to the general case. Instead, we further localize the action, reembedding the gauge subgroup $\calG_{\Gamma,\bfm}$ in the group $H = \Im(\Psi)$, where $\Psi$ is the map described in Section  \ref{sec: gauge section}. This will allow us to use results on the stability of the action on tensor spaces which in turn guarantee the stability of the action of $\calG_{\Gamma,\bfm}$.

Since $\calG_{\Gamma,\bfm} \subseteq H$, clearly $\Stab_{\calG_{\Gamma,\bfm}} (X) \subseteq \Stab_{H} (X) $. Therefore, if the action of $H$ on the space $\Hom(W_1 \vvirg W_d , V_1 \vvirg V_d)$ is generically stable, then the action of $\calG_{\Gamma,\bfm}$ on the space $\Hom(W_1 \vvirg W_d , V_1 \vvirg V_d)$ is generically stable, as well.

We establish the following result, whose proof is immediate from the product structure of $\Hom(W_1 \vvirg W_d, V_1 \vvirg V_d)$ and of $X$.
\begin{lemma}
If $X = X_1 \ootimes X_d$ then
\[
 \Stab_{H} (X) = \bigtimes_{v \in \bfv(\Gamma)} \Stab_{H_v} (X_v).
\]
In particular, the action of $H$ on $\Hom(W_1 \vvirg W_d , V_1 \vvirg V_d)$ is generically stable if and only if for every $v$ the action of $H_v$ on $\Hom(W_v,V_v)$ is generically stable.
\end{lemma}

Now, regard $X_v \in \Hom(W_v,V_v)$ as a tensor in $V_v \otimes W_v ^* = V_v \otimes \left( \bigotimes_{e \ni v} U_e' \right)$ where $U'_e = U_e$ or $U'_e = U_e^*$ depending on whether $U_e$ of $U_e^*$ appears in $W_v$. The group $H_v$ acts trivially on $V_v$; by Lemma \ref{lemma: intersection vs direct sum}, we deduce that $\Stab_{H_v} (X_v)$ is the point-wise stabilizer of $\Im X_v (V_v^*) \subseteq  \bigotimes_{e \ni v} U_e'$. In particular, if $X_v$ is generic, by linearly, $\Stab_{H_v} (X_v)$ is the simultaneous stabilizer of $n_v$ elements of $W_v^*$.

Therefore, we are reduced to study the stability of the action of a product of special linear groups $\SL(U_1) \ttimes \SL(U_k)$ on the space $U_1 \ootimes U_k \otimes V$. The study of the stability of this action is characterized in the recent \cite{DerMak:MLEMatrix,DerMakWal:MLETensors} and in the special case where $\dim V = 1$ it is characterized in \cite{BryReiVRaa:ExistenceLocMaxEntStatesViaGIT}. In particular, the following result leads to the proof of Corollary \ref{cor: PEPS}, and to a wide range of generalizations.

\begin{proposition}\label{prop: stable if deg high enough}
Let $k \geq 3$ and consider vector spaces $U_1 \vvirg U_k,V$ with $\dim U_\alpha = m$, $\dim V = n$. The action of $\SL(U_1) \ttimes \SL(U_k) $ on $U_1 \ootimes U_k \otimes V$ is generically stable unless $(k,m,n) = (3,2,1)$.
\end{proposition}

\begin{proof}
The case $(k,m,n) = (3,2,1)$ corresponds to the action of $\SL_2 \times \SL_2 \times \SL_2$ on $\bbC^2 \otimes \bbC^2 \otimes \bbC^2$; this is not stable since 
\[
9 = \dim (\SL_2 \times \SL_2 \times \SL_2) > \dim \bbC^2 \otimes \bbC^2 \otimes \bbC^2 = 8. 
\]
Except for this case, the result follows from \cite[Theorem 1.5 (Case 4)]{DerMakWal:MLETensors}, since the inequality $m \leq \frac{1}{2} m^{k-1} n$ is always verified.
 \end{proof}

\subsection{Critical case}

We conclude this section showing that the dimension of the tensor network variety in the critical case equals the upper bound of Corollary \ref{corol: upper bound with Stab}; moreover, in this case, $\dim \Stab_{\calG_{\Gamma,\bfm}} (X) = 0$. As a consequence, via Theorem \ref{thm: reduction supercritical}, we obtain the equality in the supercritical range, which completes the proof of Corollary \ref{cor: supercritical exact}.

\begin{proposition}
Let $(\Gamma,\bfm,\bfn)$ be a supercritical tensor network. Write $N_v = \prod_{e \ni v} m_e$. Then
\[
 \dim \calTNS^\Gamma_{\bfm , \bfn} = \sum_{v \in \bfv(\Gamma)} n_v N_v - d +1  - \sum_{e \in \bfe(\Gamma)} (m_e^2 - 1).
\]
\end{proposition}
\begin{proof}
 First consider the critical case, that is $N_v = n_v$. In this case, a generic $X_v \in \Hom(W_v,V_v)$ is invertible. Therefore
 \begin{align*}
  \dim \calTNS^\Gamma_{\bfm,\bfn} = &\dim G(W_1 \vvirg W_d) \cdot T(\Gamma,\bfm) = \\ 
  &\dim G(W_1 \vvirg W_d) - \dim \calG_{\Gamma,\bfm} = \sum_{v \in \bfv(\Gamma)} N_v^2  - d +1  - \sum_{e \in \bfe(\Gamma)} (m_e^2 - 1).  
 \end{align*}
  In the supercritical case, we apply Theorem \ref{thm: reduction supercritical}. Write $\bfN = (N_v : v \in \bfv(\Gamma))$, so that the tensor network $(\Gamma,\bfm,\bfN)$ is critical. Then 
  \begin{align*}
   \dim \calTNS^\Gamma_{\bfm,\bfn} =    &\dim \calTNS^\Gamma_{\bfm,\bfN} + \sum_{v \in \bfv(\Gamma)} N_v (n_v - N_v) = \\
 & \sum_{v \in \bfv(\Gamma)} N_v^2  - d +1  - \sum_{e \in \bfe(\Gamma)} (m_e^2 - 1) + \sum_{v \in \bfv(\Gamma)} N_v (n_v - N_v) = \\
 & \sum_{v \in \bfv(\Gamma)} n_v N_v - d +1  - \sum_{e \in \bfe(\Gamma)} (m_e^2 - 1). \qedhere 
  \end{align*}
\end{proof}

\section{Analysis of small cases}\label{sec: small cases}

In this section, we analyze few cases of tensor network varieties for small graphs and small bond dimension.

If $\Gamma$ only contains two vertices, then the tensor network variety is easily described as a variety of matrices whose rank is bounded from above by the bond dimension of the unique edge.

We start our analysis with the case of three vertices.

\subsection{Triangular graph}

The graph tensor associated to the triangular graph is the matrix multiplication tensor. This is the object of a rich literature, devoted to determining the value of the exponent of matrix multiplication. We refer to \cite{BlaserNotes,Lan:GeometryComplThBook} for an overview on the subject.

Let $C_3$ be the triangular graph. Write $\{1,2,3\}$ for the three vertices and $m_{12},m_{23},m_{31}$ for the three bond dimensions and $(n_1,n_2,n_3)$ for the three local dimensions, ordered as follows:
\[
\begin{tikzpicture}[scale=1.5]
\draw (0,1.5)-- (0,0);
\draw (0,0)-- (1.5,0);
\draw (0,1.5)-- (1.5,0);
\draw (.75,0) node[anchor=north] {$m_{12}$};
\draw (0.75,0.75) node[anchor=south west] {$m_{23}$};
\draw (0,.75) node[anchor=east] {$m_{31}$};
\draw[fill=black] (0,0) circle (.1cm);
\draw[fill=black] (1.5,0) circle (.1cm);
\draw[fill=black] (0,1.5) circle (.1cm);
\draw (0,1.6) node[anchor=south] {$n_3$};
\draw (1.6,0) node[anchor=west] {$n_2$};
\draw (0,0) node[anchor=north east] {$n_1$};
\end{tikzpicture}
\]

If $\bfm = (m_{12},m_{23},m_{31}) = (a,b,1)$ (in other words, the edge $\{3,1\}$ is erased) then every tensor in $W_1 \otimes W_2 \otimes W_3$ is a restriction of the graph tensor. In particular, if $\bfn = (n_1,n_2,n_3)$ with $n_1 \leq a$, $n_2 \leq ab$, $n_3 \leq b$, then
\[
 \calTNS^{C_3}_{\bfm,\bfn} = V_1 \otimes V_2 \otimes V_3.
\]
Therefore, the first interesting case is the one with bond dimensions $\bfm = (2,2,2)$. We record the cases in the subcritical range in Table \ref{table: bond 222}. For each of these cases, the lower bound for the dimension is obtained computing explicitly the rank of the differential of the parametrization map $\Phi$ of Section \ref{section: def and prelim} at a random point. We perform this calculation in Macaulay2 \cite{M2}. The scripts performing the calculation are available at \link.

\begin{table}
\[
\begin{array}{cc|c|c}
&\bfn & \text{lower bound} & \text{upper bound }  \\ \midrule 
&(2,2,2) & 8 & 8 \\
&(2,2,3) & 12 & 12 \\
&(2,2,4) & 16 & 16 \\
&(2,3,3) & 18 & 18 \\
{}^*&(2,3,4) & 22 & 24 \\
{}^*&(2,4,4) & 26 & 29 \\
&(3,3,3) & 25 & 25 \\
&(3,3,4) & 29 & 29 \\
&(3,4,4) & 31 & 31 \\
&(4,4,4) & 37 & 37 
\end{array} \]
\caption{Upper and lower bound for $\dim \calTNS^{C_3}_{\bfm,\bfn}$. The lower bound is obtained via a direct calculation. The upper bound is the value obtained in Corollary \ref{cor: MPS}. In the cases marked with ${}^*$ the two bounds do not coincide.}\label{table: bond 222}
\end{table}
Since the point to compute the differential is chosen at random, we are confident that the number recorded as a lower bound is equal to the actual dimension of the tensor network variety $\calTNS^{C_3}_{\bfm,\bfn}$. However, from a formal point of view, the sole calculation of the rank of the differential at a random point does not provide a complete proof.

The only cases where the lower bound does not match the upper bound given in Corollary \ref{cor: MPS} are the ones with $\bfn = (2,3,4)$ and $\bfn = (2,4,4)$. In these cases, we prove that the dimension of the tensor network variety equals the lower bound of Table \ref{table: bond 222}. We provide the following result, that we prove in general and will be used in Theorem \ref{thm: 234 and 244} in the cases $(a,b,r) = (3,4,2)$ and $(a,b,r) = (4,4,2)$.
\begin{lemma}\label{lemma: variety of secant lines}
Let $V_1,V_2,V_3$ be vector spaces with $\dim V_1 = 2$, $\dim V_2 = a$, $\dim V_3 = b$. Let $\sigma_r \subseteq \bbP ( V_2 \otimes V_3)$ be the variety of elements of rank at most $r$. Define 
 \[
  \calZ_{a,b,r} = \bar{\Bigl \{ T \in V_1 \otimes V_2 \otimes V_3 : T(V_1^*) \cap \sigma_r \text{ contains at least two points} \Bigr\}} \subseteq \bbP (V_1 \otimes V_2 \otimes V_3)
 \]
Then $\calZ_{a,b,r}$ is an irreducible variety and 
\[
 \dim \calZ_{a,b,r} = 2r(a+b-r)+1 .
\]
\end{lemma}
\begin{proof}
 Define the variety of secant lines
 \[
  \calS_{a,b,r} = \bar{\Bigl \{ L \in \bfG(2,V_2 \otimes V_3) : \bbP L \cap \sigma_r \text{ contains at least two points} \Bigr\}} \subseteq \bfG(2,V_2 \otimes V_3),
 \]
 where $\bfG(2,V_2 \otimes V_3)$ denotes the Grassmannian of $2$-planes in $V_2 \otimes V_3$.
 
Then $\calS_{a,b,r}$ is an irreducible variety of dimension $2 \dim \sigma_r = 2[r(a+b-r)-1]$ \cite[Section 10.3]{EisHar:3264}.

The variety $\calZ_{a,b,r}$ is an $\SL(V_1)$-bundle on $\calS_{a,b,r}$. This guarantees that $\calZ_{a,b,r}$ is irreducible and provides $\dim \calZ_{a,b,r} = \dim \calS_{a,b,r} +3 = 2[r(a+b-r)-1]+3 = 2r(a+b-r)+1$ as desired.
\end{proof}

We can characterize some small instances of $\calTNS^{C_3}_{\bfm,\bfn}$ in terms of the varieties $\calZ_{a,b,r}$ introduced in Lemma \ref{lemma: variety of secant lines}. If $Y \subseteq \bbP W$ is a projective variety, let $\hat{Y}$ denote its affine cone in the space $W$.

\begin{theorem}\label{thm: 234 and 244}
Let $\bfm =(2,2,2)$:
\begin{itemize}
 \item if $\bfn=(2,3,4)$ then $\calTNS^{C_3}_{\bfm,\bfn} = \hat{\calZ}_{3,4,2}$; in particular $\dim \calTNS^{C_3}_{\bfm,\bfn} = 22$;
 \item if $\bfn=(2,4,4)$ then $\calTNS^{C_3}_{\bfm,\bfn} = \hat{\calZ}_{4,4,2}$; in particular $\dim \calTNS^{C_3}_{\bfm,\bfn} = 26$.
\end{itemize}
\end{theorem}
\begin{proof}
The lower bound on the dimension follows from Table \ref{table: bond 222}.

By Lemma \ref{lemma: variety of secant lines}, we have $\calZ_{3,4,2} = 4 \cdot (3+4-2)+1 = 21$ and $\calZ_{4,4,2} = 4 \cdot (4+4-2)+1 = 25$. In the rest of the proof, we show that $\calTNS^{C_3} _{\bfm,(2,3,4)} \subseteq \hat{\calZ}_{3,4,2}$ and $\calTNS^{C_3} _{\bfm,(2,4,4)} \subseteq \hat{\calZ}_{4,4,2}$. 

Fix generic $X_1 ,X_2,X_3$ with $X_j \in \Hom(W_j , V_j)$ and let $T = X_1 \otimes X_2 \otimes X_3 (T(C_3,\bfm))$. Let $L = T(V_1^*) \subseteq V_2 \otimes V_3$. It suffices to show that $\bbP L\cap \sigma_2$ contains at least two points in the two cases of interest.

We are free to normalize the linear maps $X_1,X_2,X_3$ via the action of the gauge subgroup in $\GL(W_1,W_2,W_3)$ and the action of $\GL(V_1) \times \GL(V_2)\times \GL(V_3)$ on $V_1 \otimes V_2 \otimes V_3$.

Identify $X_1$ with a $2 \times 2$ matrix $B_1(v_1^{(1)},v_2^{(1)} )$ whose entries are linear combinations of the elements of a basis $\{ v_1^{(1)} , v_2^{(2)}\}$ of $V_1$ and similarly for $X_2$ and $X_3$. In this way
\[
 X_1 \otimes X_2 \otimes X_3(T(C_3,\bfm)) = \trace \biggl( B_1(v_1^{(1)},v_2^{(1)} ) \cdot B_2(v_1^{(2)}\vvirg v_3^{(2)} )\cdot B_3(v_1^{(3)} \vvirg v_4^{(3)} ) \biggr) .
\]
Write $B_1(v_1^{(1)} , v_2^{(1)}) = B_1^1 v_1^{(1)} + B_1^2 v_2^{(1)}$ and similarly for the other matrices.

By genericity, the map $X_3$ is invertible: using the action of $\GL(V_3)$, we may assume 
\[
B_{3}^{1} = \left(\begin{smallmatrix} 1 & 0 \\ 0 & 0 \end{smallmatrix} \right) \quad
B_{3}^{2} = \left(\begin{smallmatrix} 0 & 1\\ 0 & 0 \end{smallmatrix} \right) \quad
B_{3}^{3} = \left(\begin{smallmatrix} 0 & 0 \\ 1 & 0 \end{smallmatrix} \right) \quad
B_{3}^{4} = \left(\begin{smallmatrix} 0 & 0 \\ 0 & 1 \end{smallmatrix} \right) .
\]
Moreover, the linear space $\langle B_1^1 , B_1^2 \rangle$ contains at least one matrix of rank $1$; using the action of $\GL(V_1)$ and of the gauge group, we may assume $B^1_1 = \left( \begin{smallmatrix} 1 & 0 \\ 0 & 0 \end{smallmatrix}\right)$. 

With these normalizations, it is possible to verify that the line $\bbP (T(V_1^*))$ contains two rank two matrices. We provide a Macaulay2 script determining the intersection $\bbP (T(V_1^*)) \cap \sigma_2$ at \link.

If $\bfn = (2,3,4)$, this shows $\calTNS^{C_3}_{\bfm,\bfn} \subseteq \hat{\calZ}_{3,4,2}$; if $\bfn = (2,4,4)$, this shows $\calTNS^{C_3}_{\bfm,\bfn} \subseteq \hat{\calZ}_{4,4,2}$.

Finally, since $\calTNS^{C_3}_{\bfm,\bfn} \subseteq \hat{\calZ}_{3,4,2}$ and they are both irreducible varieties of dimension $22$, equality holds. Similarly, equality holds in the inclusion $\calTNS^{C_3}_{\bfm,\bfn} \subseteq \hat{\calZ}_{4,4,2}$.
\end{proof}

\subsection{Square graph}

Consider the square graph $C_4$ with local dimensions $\bfn = (n_1 \vvirg n_4)$ and bond dimensions $\bfm = (m_{12},m_{23},m_{34},m_{41})$.
\[
\begin{tikzpicture}[scale=1.5]
\draw (0,1.5)-- (0,0);
\draw (0,0)-- (1.5,0);
\draw (0,1.5)-- (1.5,1.5);
\draw (1.5,0)-- (1.5,1.5);
\draw (.75,0) node[anchor=north] {$m_{12}$};
\draw (1.55,.76) node[anchor=west] {$m_{23}$};
\draw (.75,1.55) node[anchor=south] {$m_{34}$};
\draw (0,.75) node[anchor=east] {$m_{41}$};
\draw[fill=black] (0,0) circle (.1cm);
\draw[fill=black] (1.5,0) circle (.1cm);
\draw[fill=black] (0,1.5) circle (.1cm);
\draw[fill=black] (1.5,1.5) circle (.1cm);
\draw (0,1.5) node[anchor=south east] {$n_4$};
\draw (1.5,1.5) node[anchor=south west] {$n_3$};
\draw (1.5,0) node[anchor=north west] {$n_2$};
\draw (0,0) node[anchor=north east] {$n_1$};
\end{tikzpicture}
\]

We focus on the case where all bond dimensions are equal to $2$. As in the previous section, we record in Table \ref{table: bond 2222} the lower bound obtained computing the differential of the parametrization at a random point and the upper bound obtained via Corollary \ref{cor: MPS}. As before, because of the random choice of point, we are confident that the value recorded as lower bound coincides with the value of $\dim \calTNS^{C_4}_{\bfm,\bfn}$. We provide a formal proof for the case $\bfn = (2,2,2,2)$ in Theorem \ref{thm: 2222}.

\begin{table}
\[
\begin{array}{cc|c|c}
& \bfn & \text{lower bound} & \text{upper bound}  \\ \midrule 
{}^* & (2,2,2,2) & 15 & 16 \\ 
{}^*&      (2,2,2,3) & 20 & 21 \\ 
{}^*&      (2,2,2,4) & 24 & 25 \\ 
&      (2,2,3,3) & 25 & 25 \\ 
&      (2,2,3,4) & 29 & 29 \\ 
&      (2,2,4,4) & 33 & 33 \\ 
{}^*&      (2,3,2,3) & 24 & 25 \\ 
{}^*&      (2,3,2,4) & 28 & 29 \\ 
&      (2,3,3,3) & 29 & 29 \\ 
&      (2,3,3,4) & 33 & 33 \\ 
&      (2,3,4,3) & 33 & 33 \\ 
&      (2,3,4,4) & 37 & 37 \\ 
{}^*&      (2,4,2,4) & 32 & 33 \\ 
&      (2,4,3,4) & 37 & 37 \\ 
&      (2,4,4,4) & 41 & 41 \\ 
&      (3,3,3,3) & 33 & 33 \\ 
&      (3,3,3,4) & 37 & 37 \\ 
&      (3,3,4,4) & 41 & 41 \\ 
&      (3,4,3,4) & 41 & 41 \\ 
&      (3,4,4,4) & 45 & 45 \\ 
&      (4,4,4,4) & 49 & 49 \\ 
      \end{array}
\]
\caption{Upper and lower bound for $\dim \calTNS^{C_4}_{\bfm,\bfn}$. The lower bound is obtained via a direct calculation. The upper bound is the value obtained in Corollary \ref{cor: MPS}. In the cases marked with ${}^*$ the two bounds do not coincide.}\label{table: bond 2222}
\end{table}

\begin{theorem}\label{thm: 2222}
Let $\bfm = (2,2,2,2)$ and $\bfn = (2,2,2,2)$. Then 
\[
 \dim \calTNS^{C_4}_{\bfm,\bfn} = 15;
\]
more precisely $\calTNS^{C_4}_{\bfm,\bfn} $ is a hypersurface of degree $6$.
\end{theorem}
\begin{proof}
The lower bound $\dim \calTNS^{C_4}_{\bfm,\bfn} \geq 15$ is obtained in Table \ref{table: bond 2222}.

Since $\dim V_1 \otimes V_2 \otimes V_3 \otimes V_4 = 16$, we obtain that either $\calTNS^{C_4}_{\bfm,\bfn}$ is the entire space or it is a hypersurface. 

We determine an irreducible equation of degree $6$ vanishing on $\calTNS^{C_4}_{\bfm,\bfn}$. 

This equation is a degree $6$ invariant for the action of $\GL(V_1) \ttimes \GL(V_4)$ on $V_1 \ootimes V_4$. Its construction is described explicitly in \cite{LuqThi:PolyInv4qubits,HolLuqThi:GeomDescrEntangAuxVars}. The evaluation of the invariant is performed by a Macaulay2 script \cite{M2} available at \link.

We illustrate here how to construct it and how to exploit the action of $\GL(V_1) \ttimes \GL(V_4)$ and of the gauge group to normalize the linear maps and reduce the degrees of freedom in order to allow the script evaluate the invariant.

Given a tensor $T \in V_1 \otimes V_2 \otimes V_3 \otimes V_4$, consider the bilinear map $T^{1,3}: V_1^* \times V_3^* \to  V_2 \otimes V_4$. This makes $V_2 \otimes V_4$ into a space of $2 \times 2$ matrices depending bilinearly on $V_1 \times V_3$. Let $F(T) = \det(T^{1,3})$ be the determinant (of the $2 \times 2$ matrix $V_2 \otimes V_4$) evaluated on the image of $T^{1,3}$. So $F(T)$ is a polynomial of bidegree $(2,2)$ in $V_1 \times V_3$, therefore it can be regarded as a bilinear form on $S^2 V_1 \times S^2 V_2$, where $S^2 W$ denotes the second symmetric power of a vector space $W$. Since $\dim S^2\bbC^2 = 3$, this bilinear form has an associated $3 \times 3$ matrix. The invariant $I_6$ that we are interested in is the determinant of such matrix, which is a polynomial of degree $6$ in the coefficients of the original tensor $T$.

In order to prove that $I_6$ vanishes identically on $\calTNS^{C_4}_{\bfm,\bfn}$, we apply a normalization which reduces the total degrees of freedom, then we perform the calculation symbolically in Macaulay2.

Write $T \in \calTNS^{C_4}_{\bfm,\bfn}$ as 
\[
 T = \trace\biggl( B_1(v_1^{(1)} , v_2^{(1)}) \cdots  B_4(v_1^{(4)} , v_2^{(4)}) \biggr)
\]
where $B_p(v_1^{(p)},v_2^{(p)}) = B_p^{1} v_1^{(p)} + B_p^{2} v_2^{(p)}$ are $2 \times 2$ matrices depending linearly on a fixed basis of $V_p$.

Since $\calTNS^{C_4}_{\bfm,\bfn}$ is invariant under the action of $\GL(V_1) \ttimes \GL(V_4)$ and the graph tensor is invariant under the action of the gauge subgroup, we may use these groups to normalize the matrices $B_p^{j}$. In particular, by the action of $\GL(V_1)$ and $\GL(V_3)$, we may assume $B_1^1$ and $B^1_3$ are rank one matrices; further, using the action of the gauge subgroup, we may assume $B^{1}_1 = B^1_3 = \left( \begin{smallmatrix} 1 & 0 \\ 0 & 0 \end{smallmatrix}\right)$.

With this normalization, the evaluation of the invariant is performed and we can verify that $I_6(T) = 0$ whenever $T \in \calTNS^{C_4}_{\bfm,\bfn}$. 

Since $I_6$ is irreducible, we conclude $\calTNS^{C_4}_{\bfm,\bfn}$ is a hypersurface of degree $6$.
\end{proof}

If $d = 5,6,7$, the calculation of the differential at a random point shows that in the case of constant bond dimension $2$ the dimension of tensor network varieties coincides with the upper bound of Corollary \ref{cor: MPS}. Therefore, we propose the following conjecture:
\begin{conjecture}\label{conjecture: bond 2}
 Let $d \geq 3$, $\bfm = (2 \vvirg 2)$ and $\bfn = (n_1 \vvirg n_d)$ with $n _j \geq 2$. Then 
 \[
\dim \calTNS^{C_d} _{\bfm,\bfn} = \min \left\{ 4\left(\textsum_{j=1}^d n_j - d \right)  + 1 , \textprod_{j=1}^d n_j\right\}
 \]
 except in the following cases:
 \begin{itemize}
  \item if $d = 3$: $\bfn = (2,n_2,n_3)$, with $n_2 \geq 3$, $n_3 \geq 4$ and their cyclic permutations; 
  \item if $d = 4$: $\bfn = (2,n_2,2,n_4)$ with $n_2,n_4 \geq 2$ and their cyclic permutations.
 \end{itemize}
\end{conjecture}
The results of this section, together with Theorem \ref{thm: reduction supercritical}, confirm Conjecture \ref{conjecture: bond 2} for $d=3$. As mentioned above, a direct calculation confirms the conjecture for $d=5,6,7$. In the case $d=4$, the conjecture is confirmed in the case $\bfn= (2,2,2,2)$, in all cases where the upper and lower bounds coincide in Table \ref{table: bond 2222} and in the supercritical cases constructed from those.

\subsection*{Acknowledgements} A.B. and C.D.L. acknowledge support from GNSAGA of INDAM. We would like to thank Giorgio Ottaviani and Daniel Stilck Fran\c{c}a for useful comments and suggestions.

{\small
\bibliographystyle{alpha}
\bibliography{mamu.bib}
}

\end{document}